\documentclass[11pt]{scrarticle}

\usepackage[natbibapa]{apacite}
\usepackage[utf8]{inputenc}
\usepackage{bm}
\usepackage{tabu, booktabs}
\usepackage[pdfencoding=auto, psdextra]{hyperref}
\hypersetup{
    colorlinks,
    linkcolor={blue!70!black},
    citecolor={blue!70!black},
    urlcolor={blue!70!black}
}

\IfFileExists{xurl.sty}{%
    \usepackage{xurl}
}{%
    \usepackage{url}%
}

\usepackage{environ}

\AtBeginDocument{%
    \providecommand{\doi}[1]{#1}%
    \renewcommand{\doi}[1]{\href{https://doi.org/#1}{\nolinkurl{https://doi.org/#1}}}%
}

\makeatletter
\newif\ifttm@pending@url
\newcommand{\ttm@pending@url@date}{}
\newcommand{\ttm@pending@url@body}{}
\newcommand{\ttm@print@pending@url}{%
    \ifttm@pending@url
        \ifx\ttm@pending@url@date\@empty
            \BRetrievedFrom
        \else
            \BRetrieved{\ttm@pending@url@date}%
        \fi
        \ttm@pending@url@body
        \global\ttm@pending@urlfalse
    \fi
}
\RenewEnviron{APACrefURL}[1][]{%
    \gdef\ttm@pending@url@date{#1}%
    \global\let\ttm@pending@url@body\BODY
    \global\ttm@pending@urltrue
}
\RenewEnviron{APACrefDOI}{%
    \global\ttm@pending@urlfalse
    \BODY
}
\AtBeginDocument{%
    \let\ttm@original@PrintBackRefs\PrintBackRefs
    \renewcommand{\PrintBackRefs}[1]{%
        \ttm@print@pending@url
        \ttm@original@PrintBackRefs{##1}%
    }%
    \let\ttm@original@endthebibliography\endthebibliography
    \renewcommand{\endthebibliography}{%
        \ttm@print@pending@url
        \ttm@original@endthebibliography
    }%
}
\makeatother

\usepackage{csquotes}
\usepackage{authblk} 
\usepackage{setspace}
\usepackage[
    font=small,
    labelfont={bf,sf},
    labelsep=quad,
    format=plain,
]{caption}

\usepackage[left=1in, top=1in, right=1in, bottom=1in]{geometry}
\usepackage{graphicx,bm,colonequals,amsmath,amssymb,amsthm,url,xcolor,bbm}
\usepackage{array,tabularx,multirow, booktabs, makecell}

\usepackage{enumitem}
\usepackage{placeins} 
\usepackage{gensymb} 

\graphicspath{{plots/}}

\usepackage{xr}
\makeatletter
\newcommand*{\addFileDependency}[1]{
    \typeout{(#1)}
    \@addtofilelist{#1}
    \IfFileExists{#1}{}{\typeout{No file #1.}}
}
\makeatother

\usepackage{mathtools}

\usepackage[ruled,vlined]{algorithm2e}
\SetKwInput{KwInput}{Input}
\usepackage{algorithmic}
\usepackage{array,framed}
\usepackage{float}

\makeatletter
\renewcommand{\paragraph}{%
    \@startsection{paragraph}{4}%
    {\z@}{2ex \@plus 1ex \@minus .2ex}{-1em}%
    {\normalfont\normalsize\bfseries}%
}
\makeatother

\DeclareMathAlphabet\mathbfcal{OMS}{cmsy}{b}{n}

\newcommand{\keywords}[1]{%
    \vspace{1em}
    \noindent\textbf{Keywords: }#1
}

\usepackage{tabularx}
\usepackage{array}

\newcolumntype{L}{>{\raggedright\arraybackslash}X} 
\newcolumntype{C}[1]{>{\centering\arraybackslash}p{#1}} 
\newcolumntype{K}[1]{>{\raggedright\arraybackslash}p{#1}} 

\newcommand{\bss}{\boldsymbol{s}}
\newcommand{\bst}{\boldsymbol{t}}

\newcommand{\bsy}{\boldsymbol{y}}
\newcommand{\bsz}{\boldsymbol{z}}

\newcommand{\bsQ}{\boldsymbol{Q}}

\newcommand{\bfu}{\mathbf{u}}
\newcommand{\bfv}{\mathbf{v}}

\newcommand{\bfB}{\mathbf{B}}

\newcommand{\bfI}{\mathbf{I}}

\newcommand{\bfK}{\mathbf{K}}
\newcommand{\bfL}{\mathbf{L}}

\newcommand{\bfS}{\mathbf{S}}

\newcommand{\bfU}{\mathbf{U}}

\newcommand{\bfW}{\mathbf{W}}

\newcommand{\bfY}{\mathbf{Y}}
\newcommand{\bfZ}{\mathbf{Z}}

\newcommand{\sfT}{\mathsf{T}}

\newcommand{\bsbeta}{\boldsymbol{\beta}}
\newcommand{\bsgamma}{\boldsymbol{\gamma}}

\newcommand{\bszeta}{\boldsymbol{\zeta}}

\newcommand{\bstheta}{\boldsymbol{\theta}}

\newcommand{\bfone}{\mathbf{1}}
\newcommand{\bfzero}{\mathbf{0}}

\newcommand{\calG}{\mathcal{G}}
\newcommand{\calH}{\mathcal{H}}
\newcommand{\calI}{\mathcal{I}}

\newcommand{\calK}{\mathcal{K}}

\newcommand{\calN}{\mathcal{N}}

\newcommand{\calO}{\mathcal{O}}
\newcommand{\calP}{\mathcal{P}}

\newcommand{\calS}{\mathcal{S}}
\newcommand{\calT}{\mathcal{T}}

\newcommand{\bbE}{\mathbb{E}}

\newcommand{\bbP}{\mathbb{P}}

\newcommand{\bbR}{\mathbb{R}}

\newcommand{\Ndim}{{L}}
\newcommand{\Nobs}{{N}}
\newcommand{\Ntrain}{{N_{\text{train}}}}

\usepackage{xparse}

\NewDocumentCommand{\Ndime}{o}{%
  L\IfValueT{#1}{^{#1}}%
}
\NewDocumentCommand{\Nobse}{o}{%
  N\IfValueT{#1}{^{#1}}%
}

\newtheorem{theorem}{Theorem}
\newtheorem{proposition}{Proposition}
\newtheorem{lemma}{Lemma}

\DeclareMathOperator*{\argmax}{arg\,max\ }

\title{%
\vspace{-1cm}
Data-Efficient Generative Modeling of Non-Gaussian Global Climate Fields via Scalable Composite Transformations}
\date{}

\author[1]{Johannes Brachem\thanks{\texttt{brachem@uni-goettingen.de}}}
\author[2]{Paul F.V. Wiemann}
\author[3]{Matthias Katzfuss}

\affil[1]{Chair of Statistics, University of Göttingen, Germany}
\affil[2]{Department of Statistics, The Ohio State University, USA}
\affil[3]{Department of Statistics, University of Wisconsin--Madison, USA}

\begin{document}

\maketitle
\begin{abstract}
    \begin{center}
        \textbf{Abstract}
        \vspace{0.5em}
    \end{center}
    \noindent
    Quantifying uncertainty in climate-model output requires characterizing internal variability, often through large ensembles of physical climate-model runs. Since each additional ensemble member is computationally expensive, only limited numbers of replicated fields are typically available under a fixed model configuration and forcing scenario. We propose a data-efficient stochastic generator for the internal variability of global climate fields, specifically designed to overcome these sample-size constraints. The framework targets the distribution of climate-model output under fixed forcing and model physics. Inspired by copula modeling, our approach constructs a highly expressive joint distribution via a composite transformation to a multivariate standard normal space. We combine a nonparametric Bayesian transport map for spatial dependence modeling with flexible, spatially varying marginal models, essential for capturing non-Gaussian behavior and heavy-tailed extremes. These marginals are defined by a parametric model followed by a semi-parametric B-spline correction to capture complex distributional features. The marginal parameters are spatially smoothed using Gaussian-process priors with low-rank approximations, rendering the computational cost linear in the spatial dimension. When applied to global log-precipitation-rate fields under fixed forcing at more than 50,000 grid locations, our stochastic surrogate achieves high fidelity, accurately quantifying the climate distribution's spatial dependence and marginal characteristics, including the tails. Using only 10 training fields, it outperforms a state-of-the-art competitor trained on 80 fields, effectively octupling the computational budget for climate research. We provide a Python implementation at \url{https://github.com/jobrachem/ppptm}.

\end{abstract}


\clearpage

\section{Introduction}
\label{sec:introduction}

\subsection{Background and motivation}
\label{sec:intro-background}

Earth system models (ESMs) are central tools for studying the climate system and for assessing the risks associated with climate variability and change. A single ESM run, however, represents only one possible realization of the climate system. Even under a fixed model configuration and forcing scenario, small perturbations to the initial state can lead to different weather sequences, regional trends, and extremes. This internal variability is not merely noise or model error, but an important part of the climate distribution; it can mask or amplify forced changes over policy-relevant periods and affects the probability of spatially coherent extremes. Characterizing it requires ensembles of simulations. Their cost, however, is substantial: each additional ensemble member requires running the full physical model and storing another high-dimensional spatio-temporal output. This creates a strong motivation for statistical surrogates that can reproduce aspects of the distribution of climate-model output at a much lower computational and storage cost.

\subsection{Emulators, stochastic generators, and related surrogates}
\label{sec:intro-emulators}

The term emulator is used broadly for computationally efficient substitutes for climate-model output \citep{Tebaldi2025-EmulatorsClimateModel}. A useful distinction is between surrogates that learn how climate-model output changes across inputs, such as emissions scenarios, forcing levels, physical parameters, or boundary conditions, and stochastic generators that target the distribution of output under a fixed model configuration and forcing scenario. The latter aim to sample additional realizations of internal variability after observing only a limited ensemble of expensive model runs \citep{Castruccio2019-ReproducingInternalVariability,Edwards2019-MultivariateGlobalSpatiotemporal,Hu2021-ApproximatingInternalVariability}. Our proposed model belongs to this second class: it is a data-efficient stochastic generator for high-dimensional climate fields.

Prior stochastic-generator work has reproduced regional temperature-trend variability from small ensembles \citep{Castruccio2019-ReproducingInternalVariability}, developed multivariate global spatio-temporal generators for climate ensembles \citep{Edwards2019-MultivariateGlobalSpatiotemporal}, and used spatial stochastic generators for bias-corrected temperature projections \citep{Hu2021-ApproximatingInternalVariability}. Related climate-surrogate work has also connected stochastic modeling to storage and computational savings: \cite{Huang2023-SavingStorageClimate} study model-based stochastic approximations as a route to saving computational time and storage in climate ensembles, while \cite{Song2024-EfficientStochasticGenerators} propose spherical-harmonic stochastic generators that rapidly emulate high-resolution CESM2-LENS2 temperature simulations with reduced computational and storage costs. These developments connect to the longer stochastic-weather-generator tradition \citep{Ailliot2015-StochasticWeatherGenerators}, but target large climate-model ensembles and internal variability rather than site-level observed weather sequences.

Machine-learning emulators form another rapidly developing branch of the landscape. Data-driven weather-forecasting systems such as AIFS \citep{Lang2024-AIFSDataDrivenForecasting} and climate emulators such as ACE2 \citep{Watt-Meyer2025-ACE2AccuratelyLearning} have shown strong performance in high-frequency spatio-temporal prediction and in simulations that remain stable over many repeated forecast steps. Diffusion and related generative models provide a particularly relevant modern comparison, with applications to forecast-ensemble generation \citep{Li2024-GenerativeEmulationWeathera}, continuous probabilistic weather-forecast trajectories \citep{Andrae2025-ContinuousEnsembleWeather}, and probabilistic downscaling of ESM fields \citep{Hess2025-FastScaleadaptiveUncertaintyaware}. More broadly, this literature includes normalizing flows \citep{Kobyzev2021-NormalizingFlowsIntroduction}, spatial generative models \citep{Ng2022-SphericalPoissonPoint,Ng2023-MixtureModelingNormalizing}, and GAN-based climate-model emulators \citep{Ayala2021-LooselyConditionedEmulation,Besombes2021-ProducingRealisticClimate}.

These machine-learning models address important but partly different targets, and are typically trained on extensive high-frequency weather records, forecast-ensemble collections, or high-resolution paired datasets. Many learn short-step forecasting dynamics or conditional mappings. By contrast, we estimate the distribution of a high-dimensional spatial field from a small number of replicated climate-model fields representing a specified target distribution (under fixed forcing in our application), without learning a time-evolution model. This is useful when a conventional climate model has already been run for a target climate state or scenario, such as a future time slice, but only a limited ensemble is available: the fitted stochastic generator can then produce additional plausible spatial fields for distributional summaries, uncertainty quantification, and tail-risk analyses around that target. Thus, our approach complements machine-learning emulators, especially when only tens of expensive ensemble replicates are available and a statistically explicit distribution over spatial fields is desired.

\subsection{Modeling target and data regime}
\label{sec:intro-target}

Formally, we consider a non-Gaussian spatial field $\bsy = (y_1, \dots, y_{\Ndim})^\sfT$ observed through training samples $\bsy_1, \dots, \bsy_{\Nobs}$ from a climate model, where $\Ndim$ may be very large and $\Nobs$ is small. In our application, $\Nobs \leq 98$ global log-precipitation-rate fields from the Community Earth System Model (CESM) Large Ensemble Project \citep{Kay2015-CommunityEarthSystem} are observed on a roughly 1\degree\ longitude-latitude grid with $\Ndim = 288 \times 192 = 55{,}296$ locations before trimming duplicate polar grid points. We use daily averages of July 1 data from 98 consecutive CESM model years starting with model year 402 in the preindustrial control simulation. Over this period, the forcing is fixed and no systematic climate-change signal is expected, so we treat the fields as approximately iid replicates from a common distribution and do not model their temporal ordering, temporal autocorrelation, seasonal dynamics, or response to time-varying covariates.

This restriction is intentional and defines the contribution of our paper. We focus on the difficult problem of learning a flexible distribution for a massive non-Gaussian spatial field when the number of replicates is on the order of $10$ to $100$. Such fields can exhibit skewness, heavy tails, local multimodality, spatially varying marginal behavior, and complex dependence. Standard Gaussian-process models and many related spatial methods provide valuable tools for high-dimensional spatial modeling, but they often rely on Gaussianity or on parametric covariance structure \citep{Banerjee2014-HierarchicalModelingAnalysis,Risser2016-ReviewNonstationarySpatial,Choi2013-NonparametricEstimationSpatial,Porcu2021-NonparametricBayesianModeling,Nychka2018-ModelingEmulationNonstationary,Wiens2020-ModelingSpatialData}. Copula and transformation approaches relax marginal Gaussianity and can represent richer dependence \citep{Krupskii2018-FactorCopulaModels,Mondal2024-NonStationaryFactor,Klein2022-MultivariateConditionalTransformation,Herp2024-GraphicalTransformationModels}, but existing high-dimensional implementations typically do not scale to global climate grids with more than $50{,}000$ locations and very small sample sizes. Deep generative models are highly expressive, but their training requirements and tuning costs can be substantial in precisely this low-replicate regime.

\begin{figure}[bt]
    \centering
    \includegraphics[trim=3.5mm 2mm 0 0, clip,width=.96\linewidth]{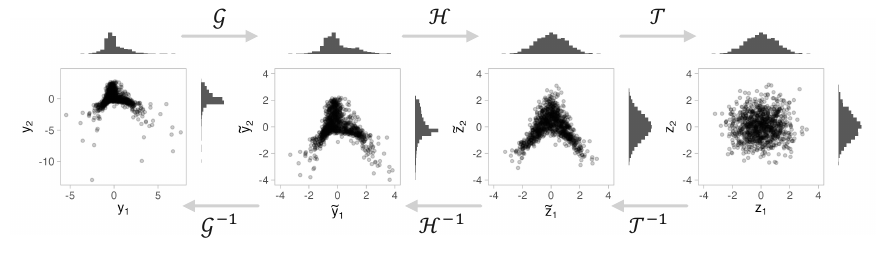} \hfill  ~~\\
    \includegraphics[width=\linewidth]{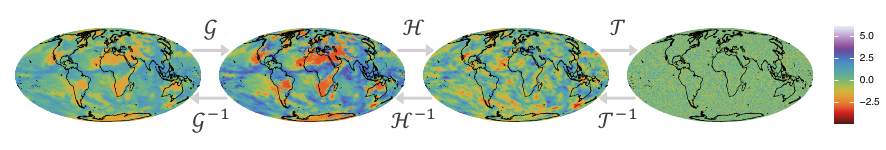}
    \caption{Top: Conceptual illustration of the three-part composite transformation mapping a bivariate non-Gaussian distribution to a standard Gaussian distribution. The parametric marginal model $\calG$ (here, a location-scale $t$-distribution) standardizes the margins and brings the tails closer to the bulk of the distribution. The semi-parametric model $\calH$ corrects residual deviations from marginal Gaussianity. Finally, the Bayesian transport map $\calT$ captures and removes the nonlinear dependence, achieving joint Gaussianity.
        Bottom: Analogous transformation of a global log-precipitation-rate field produced by a climate model (see main text). Due to the joint estimation of $\calG$ and $\calH$ here, the behavior of $\calG$ cannot be easily interpreted separately from $\calH$.
    }
    \label{fig:model-summary}
\end{figure}

\subsection{Scalable composite transformations}
\label{sec:intro-sct}

We propose scalable composite transformations (SCT) that jointly map the observed field to a multivariate standard Gaussian reference distribution through an invertible composition,
$\calT \circ \calH \circ \calG$, as illustrated in Figure~\ref{fig:model-summary}.
The first two components act elementwise on the marginal distributions. A parametric layer $\calG$ provides a stable baseline transformation and controls tail behavior through a manually selected parametric distributional family. A semi-parametric layer $\calH$ then corrects residual marginal misfit in the bulk of the distribution using a regularized monotone B-spline transformation whose parameters are smoothed across space. The final component $\calT$ is the scalable Bayesian transport map of \cite{Katzfuss2023-ScalableBayesianTransport}, which acts as a copula-like dependence model for the marginally transformed field.

The marginal layers allow the model to represent strongly non-Gaussian local behavior without forcing the transport map to absorb marginal misspecification, while the transport map captures nonlinear spatial dependence through a sequence of conditional regressions ordered from coarse to fine spatial scales. In our empirical comparisons, this structure yields accurate predictive distributions for global log-precipitation-rate fields even with far fewer training fields than competing transport-map baselines. The resulting stochastic generator can be fit and used to generate samples on an ordinary laptop, providing a practical way to generate additional plausible spatial fields for distributional summaries, uncertainty quantification, and tail-risk analyses tied to the learned target distribution.
The remainder of this article is organized as follows. Section~\ref{sec:model} introduces the composite transformation model and its three components. Section~\ref{sec:inference} describes the two-stage estimation strategy. Section~\ref{sec:application} applies the method to CESM log-precipitation-rate fields on regional and global grids and compares it to existing scalable transport-map baselines. Section~\ref{sec:conclusion} concludes with limitations and future directions. Additional proofs, computational details, compressed-storage discussion, and diagnostic figures are provided in the Appendix. An implementation is available as a Python library at \url{https://github.com/jobrachem/ppptm}, building on Liesel \citep{Riebl2023-LieselProbabilisticProgramming} and a Python implementation of Bayesian transport maps (\url{https://github.com/katzfuss-group/batram}).

\section{Model}
\label{sec:model}

\subsection{Overview}

Consider a continuous random vector $\bsy = (y_1, \dots, y_\Ndim)^\sfT$, for example representing a spatial field at $\Ndim$ locations or grid points.
We model the joint distribution of $\bsy$ using a three-part composition $(\calT \circ \calH \circ \calG)$ mapping to a multivariate standard Gaussian space:
\begin{equation}
    \label{eq:composite-map}
    (\calT \circ \calH \circ \calG)(\bsy) \sim \calN_{\Ndim}(\bfzero, \bfI),
\end{equation}
where $\calH$ and $\calG$ are spatially smoothed element-wise operators ensuring that marginally $(\calH_i\circ\calG_i)(y_i) \sim \calN(0, 1)$ for all $i = 1, \dots, \Ndim$.
For $\calG$, we construct a preliminary element-wise transformation based on a manually chosen parametric distributional family.
For $\calH$, we build on work on conditional transformation models \citep{Hothorn2018-MostLikelyTransformations} and penalized transformation models \citep{Brachem2025-BayesianPenalizedTransformation} to construct a semi-parametric second layer of element-wise transformations.
For $\calT$, we employ the triangular scalable Bayesian transport map proposed by \cite{Katzfuss2023-ScalableBayesianTransport}. Details on all parts of the composition are provided in the following sections.
We denote intermediate quantities as:
\begin{align*}
    \tilde \bsy & = \calG(\bsy)                                                &                              \\
    \tilde \bsz & = \calH(\tilde \bsy) = (\calH \circ \calG)(\bsy)             & \text{(marginally Gaussian)} \\
    \bsz        & = \calT(\tilde \bsz) = (\calT \circ \calH \circ \calG)(\bsy) & \text{(jointly Gaussian)}.
\end{align*}
Given the triangular structure of $\calT = \big(\calT_1(\tilde z_1),\calT_2(\tilde \bsz_{1:2}),\ldots,\calT_{\Ndim}(\tilde \bsz_{1:\Ndim})\big)^\sfT$ and the element-wise nature of $\calG$ and $\calH$, the joint density of $\bsy$
factorizes as:
\begin{equation}
    \label{eq:joint-density}
    p(\bsy) = \prod_{i=1}^\Ndim \left(
    \phi\bigl( \calT_i(\tilde \bsz_{1:i})\bigr)
    \times
    \left|
    \frac{\partial\, \calT_i(\tilde \bsz_{1:i})}{\partial \tilde z_i}
    \right|
    \times
    \left|
    \frac{\partial\, \calH_i(\tilde y_i)}{\partial \tilde y_i}
    \right|
    \times
    \left|
    \frac{\partial\, \calG_i(y_i)}{\partial y_i}
    \right|
    \right),
\end{equation}
where $\phi$ denotes the standard Gaussian probability density function (PDF). All three parts of the composition are invertible, so the model can be used to generate new samples by drawing $\bsz^* \sim \calN(\bfzero, \bfI)$ and computing $\bsy^* = (\calG^{-1} \circ \calH^{-1} \circ \calT^{-1})(\bsz^*)$. Closed-form inverses exist for $\calG$ and $\calT$ (see Section~\ref{sec:G} and Section~\ref{sec:transport-map-inference}), while $\calH^{-1}$ is approximated numerically.

\subsection{Copula-inspired separation of marginal and dependence models}
Our approach is grounded in the separation of marginal and joint distributions established by Sklar's theorem (\citeyear{Sklar1959-FonctionsRepartitionDimensions}), the foundation of copula models. However, we (equivalently) map to a standard Gaussian space instead of the conventional standard uniform space. This choice simplifies the development of $\calH$, offering direct control over marginal tail behavior, and it aligns with the scalable Bayesian transport map $\calT$, which targets a standard Gaussian.

In our model, the marginal distributions are characterized by the composition of $\calH$ and $\calG$. For any continuous univariate random variable $y_i$, there exists a unique strictly monotonically increasing transformation to the standard Gaussian space, $h_i = \Phi^{-1} \circ F_i$, such that $F_i(y_i) = \Phi(h_i(y_i))$ and $h_i(y_i) \sim \calN(0, 1)$, where $\Phi$ is the standard Gaussian CDF and $F_i$ is $y_i$'s CDF \citep[see Corollary 1 in][]{Hothorn2018-MostLikelyTransformations}.
We model these transformations as $h_i = \calH_i \circ \calG_i$, where $\calG_i$ is a parametric base model (see Section~\ref{sec:G}) and $\calH_i$ is a semi-parametric B-spline correction (see Section~\ref{sec:H}). Jointly applying $\calG_i$ and $\calH_i$ results in a marginally Gaussianized vector $\tilde \bsz = (\calH \circ \calG)(\bsy)$.

The dependence structure of $\tilde \bsz$ is captured by a multivariate transformation to a multivariate standard Gaussian space, $\calT: \bbR^\Ndim \rightarrow \bbR^\Ndim$, such that $\bbP(\tilde \bsz \leq \bst) = \Phi_\Ndim(\calT(\bst))$ and $\calT(\tilde \bsz) \sim \calN(\bfzero, \bfI)$, where $\Phi_\Ndim$ denotes the $\Ndim$-dimensional standard Gaussian CDF.
Without loss of generality, we can assume this transformation $\calT$ to be lower triangular, where each $\calT_i(\tilde \bsz_{1:i})$ is increasing in $\tilde z_i$ \citep{Rosenblatt1952-RemarksMultivariateTransformationa}. We adopt the parameterization $\calT_i(\tilde \bsz_{1:i}) = (\tilde z_i - f_i(\tilde \bsz_{<i})) / d_i$ with $d_i \in \bbR^+$, $f_i: \bbR^{i-1} \rightarrow \bbR$ for $i = 2, \dots, \Ndim$, and $f_i(\tilde \bsz_{<i}) \equiv 0$ for $i=1$. The triangular map $\calT$ takes the role of a copula, describing the dependence structure of the marginally transformed data $\tilde \bsz$. $\calT$ can also be understood as a series of nonlinear regressions $\tilde z_i = f_i(\tilde \bsz_{<i}) + \epsilon_i$, with $\epsilon_i \sim \calN(0, d_i^2)$.
Crucially, the triangular structure implies that the first element of $\tilde \bsz$ is modeled as Gaussian, and subsequent elements as conditionally Gaussian. For variables appearing early in the ordering sequence, the conditioning set provides limited predictive information, effectively reducing the conditional Gaussian assumption to an unconditional one. This necessitates the accurate marginal Gaussianization provided by $\calG$ and $\calH$, as the transport map cannot correct marginal deviations especially for these early-ordered variables.

\subsection{Details on the parametric transformations \texorpdfstring{$\calG$}{G}}
\label{sec:G}

The parametric transformations $\calG_i$ are defined as $\tilde y_i = \calG_i(y_i) = \Phi^{-1} \bigl( F_i(y_i | \bszeta_i) \bigr)$,
where $F_i$ is a user-specified parametric CDF governed by unknown parameters $\bszeta_i$, and $\Phi^{-1}$ is the standard Gaussian quantile function. We assume that $F_i$ admits a density. This parametric transformation is strictly monotonically increasing by construction. By the chain rule, the derivative is
$\frac{\partial}{\partial y_i} \calG_i(y_i) = F'_i(y_i | \bszeta_i)\phi(\calG_i(y_i))^{-1}$,
where $\phi$ is the standard Gaussian density and $F'_i$ is the density corresponding to $F_i$.
The inverse parametric transformation is
$y_i = \calG_i^{-1}(\tilde y_i) = F_i^{-1}(\Phi(\tilde y_i) | \bszeta_i)$.
As a concrete example, if $F_i$ is Gaussian with mean $\mu_i$ and variance $\sigma_i^2$, then  $\calG_i(y_i) = (y_i - \mu_i)/\sigma_i$ simplifies to a location-scale transformation.
For simplicity, we assume the same distributional family $F$ is chosen for all locations $i = 1, \ldots, \Ndim$.
Note that this common-family assumption applies only to the parametric base layer $\calG$: the full CDF implied by the marginal model is $\Phi(\calH_i[\Phi^{-1}(F(y_i \mid \bszeta_i))])$, so the location-specific semi-parametric correction $\calH_i$ can represent spatially varying departures from $F$ in the well-supported part of the distribution, while $F$ remains important as the regularization target and tail-extrapolation model.

When the number of training replicates $\Nobs$ is small, estimating $\bszeta_i$ independently at each spatial location is difficult. However, because geophysical processes typically vary smoothly, it is reasonable to assume spatial smoothness for the distributional parameters. To regularize estimation and share information across space, we place Gaussian-process (GP) priors over space on the parameters $\bszeta_i$. Specifically, let $\bszeta_i = [\zeta_{i,1}, \dots, \zeta_{i,P}]$, where $P$ denotes the number of parameters modeled in the parametric marginal distribution $F_i$. We define independent GP priors $\zeta_{i,p} \sim \calG \calP (0, k_{\zeta_p}(\cdot, \cdot'))$, for each $p=1, \dots, P$, where $k_{\zeta_p}$ denotes an appropriate correlation function. Each process and associated kernel $k_{\zeta_p}$ is governed by hyperparameters $\bstheta^\calG_p$, typically comprising an amplitude $\tau^\calG_p$ and a length scale $\ell^\calG_p$, such that $\bstheta^\calG_p = (\tau^\calG_p, \ell^\calG_p)^\sfT$. Note that we assume all elements of $\bszeta_i$ to be defined on the full real line. Parameters with restricted domains are therefore treated on an appropriately transformed scale.

\subsection{Details on the semi-parametric transformations \texorpdfstring{$\calH$}{H}}
\label{sec:H}
We design the semi-parametric transformations $\calH_i$ to provide flexible corrections to the parametric base model $\calG$ within the bulk of the data distribution, while seamlessly reverting to the parametric tails for robust extrapolation.
This setup deliberately shifts much of the flexible marginal modeling from the choice of the parametric family $F_i$ to the regularized correction $\calH_i$: in many applications, one can use a relatively simple and stable base family and let $\calH_i$ capture departures from that family where these departures are supported by the data.

We parameterize $\calH_i: \bbR \rightarrow \bbR$ as a strictly monotonically increasing B-spline on the interval $[k_1, k_m]$, given by
\begin{equation}
    \label{eq:onion-spline}
    \calH_i(\tilde y_i) = \sum_{j=1}^J B_j(\tilde y_i)
    \left(
    \gamma_{1,i} + \sum_{\ell = 2}^j \exp(\gamma_{\ell,i})
    \right), \qquad \text{if} \quad \tilde y_i \in [k_1, k_m],
\end{equation}
and the identity function if $\tilde y_i \in (-\infty, k_1) \cup (k_m, \infty)$. The functions $B_j$ are cubic B-spline bases defined on a sequence of $J+4$ knots
$k_{-2} < k_{-1} < \dots < k_{m+3}$, where $k_1, \dots, k_m$ are the \textit{interior knots} with $m = J - 2$. The knots are constant across all dimensions $i = 1, \dots, \Ndim$. The knot sequence is defined as an equidistant grid with $k_{j+1}-k_j = k$ denoting the constant difference between two adjacent knots, and they are set up by choosing $k_2=a$ and $k_{m-1}=b$ using Theorem~\ref{thm:tail-probabilities} (see below). The vector $\bsgamma_i \in \calK \subset \bbR^{J}$ contains an intercept parameter $\gamma_{1,i}$ and log-increments in the spline parameters, $\gamma_{2,i}, \dots, \gamma_{J,i}$, ensuring monotonicity.
To enforce the identity behavior outside $[a,b]$ and ensure $C^2$ continuity at the boundaries, we impose specific constraints on the coefficients.
To match the slope and level of the identity function, we fix the first coefficient to  $\gamma_{1,i} = k_0$ for all locations $i$, and we fix three $\gamma_{j,i} = \log(k)$ for $j \in \{2,3,4\} \cup \{J-2, J-1, J\}$ on either side of the knot basis. The remaining log-increments are parameterized via unconstrained parameters $\bsbeta_i \in \bbR^D$ ($D = J-7$):
\begin{equation}
    \label{eq:onion-coefficients}
    \gamma_{j,i} = \beta_{i,j-4} -
    \log \left(
    \frac{\sum_{d=1}^D \exp(\beta_{i,d})}{(m-5)k}
    \right), \quad \text{for } j = 5, \dots, J-3.
\end{equation}
This normalization ensures that the spline traverses the vertical distance between the fixed endpoints exactly, preserving the identity mapping constraints. To emphasize the dependence on the unconstrained parameters $\bsbeta_i$, we write $\calH_i(\cdot | \bsbeta_i)$ in the following.
Figure~\ref{fig:onion-knots} provides an overview of the parameters and knots in our modified monotonically increasing spline.
Panel~a) in Figure~\ref{fig:onion-spline} illustrates the construction of $\calH_i$ using randomly sampled $\bsbeta_i$.
We now establish key theoretical properties that justify this architecture. Proofs are given in Section~A of the Appendix \citep{Brachem2026-SCT-Appendix}.

\begin{figure}[t]
    \centering
    \includegraphics[width=\linewidth]{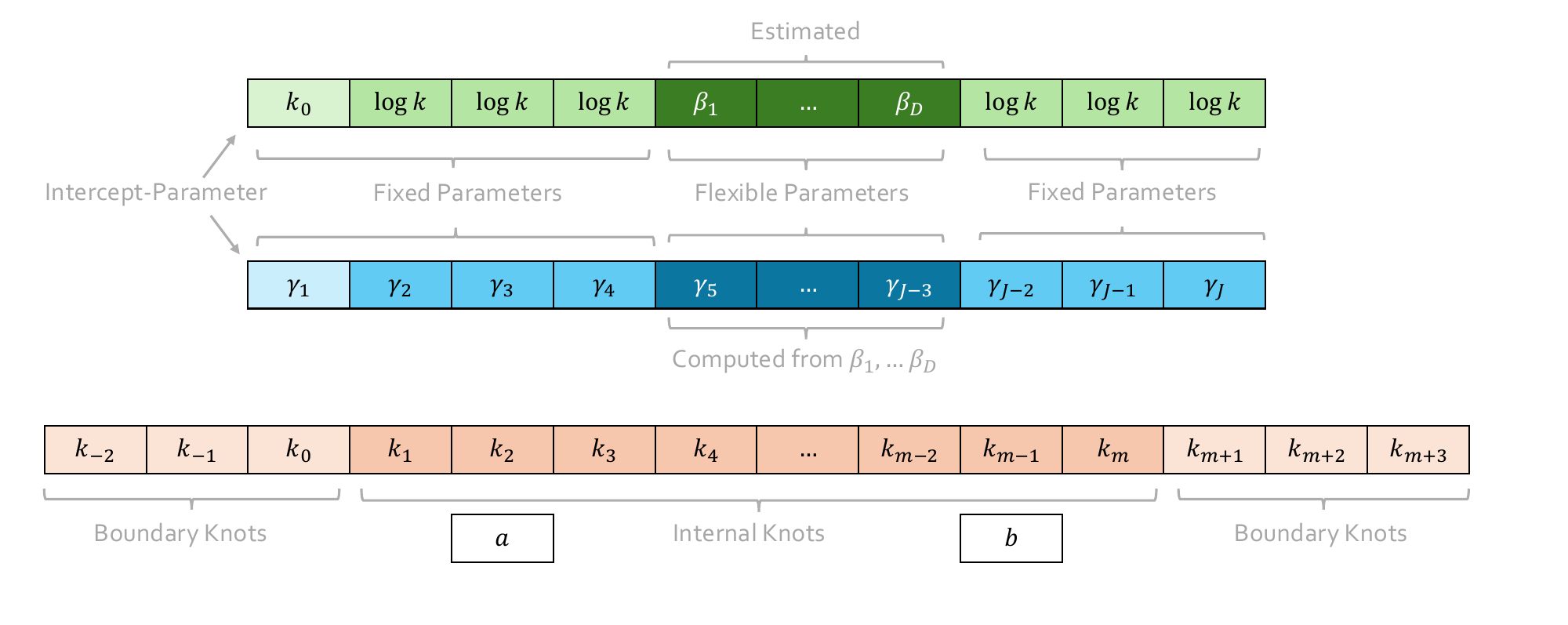}
    \caption{Illustration of how parameters and knots in our modified monotonically increasing B-spline relate to each other. The bottom panel (orange) shows the knots. The middle panel (blue) shows the spline parameters on the level of $\gamma_1, \dots, \gamma_J$, aligned with corresponding knots. The top panel (green) shows the values of the fixed $\gamma$ parameters, and aligns the freely estimated parameters $\beta_1, \dots, \beta_D$ with their counterparts $\gamma_5, \dots, \gamma_{J-3}$ (see Eq.~(\ref{eq:onion-coefficients})). Note that $k = k_{j+1}-k_j$ is the (constant) distance between two adjacent knots, and that $J =D+7$ and $m = D+5$.}
    \label{fig:onion-knots}
\end{figure}

\begin{theorem}[Reduction to base model]
    \label{thm:collapse-to-identity}
    Let $(\calH_i \circ \calG_i)(y_i) \sim \calN(0, 1)$, with $\calH_i$  as defined above
    and $\calG_i(y_i)=\Phi^{-1}(F_i(y_i))$, where $F_i: \bbR \rightarrow [0,1]$ is a continuous cumulative distribution function.
    If $\bsbeta_i = \beta_i \bfone_D$ for any constant $\beta_i \in \bbR$,
    then $\bbP(y_i \leq t) = F_i(t)$ for all $t \in \bbR$.
\end{theorem}

This theorem provides a clear mechanism for model parsimony: we can encourage the model to revert to the simpler parametric family $F_i$ by shrinking the entries of $\bsbeta_i$ towards a constant value. This mechanism is the basis for our local regularization prior introduced below in Section~\ref{sec:prior-trafo}.

\begin{theorem}[Tail probabilities]
    \label{thm:tail-probabilities}
    Let $(\calH_i \circ \calG_i)(y_i) \sim \calN(0, 1)$, with $\calH_i$ as defined above and $\calG_i(y_i)=\Phi^{-1}(F_i(y_i))$, where $F_i: \bbR \rightarrow [0,1]$ is a continuous cumulative distribution function and $\Phi^{-1}$ is the standard normal quantile function. Further, let $q_a = F_i^{-1}(\Phi(a))$ and $q_b = F_i^{-1}(\Phi(b))$ be the quantiles of the distribution with CDF $F_i$ at probability levels $\Phi(a)$ and $\Phi(b)$, respectively.
    Then $\bbP(y_i \leq t) = F_i(t)$ for all $t \in (-\infty, q_a] \cup [q_b, \infty)$.
\end{theorem}

Consequently, tail behavior is determined entirely by the parametric choice $F_i$, decoupling the estimation of the bulk (where data are abundant) from the estimation of the tails (where data are scarce by definition). This is a major conceptual improvement over the transformation function used by \cite{Brachem2025-BayesianPenalizedTransformation}, where the transformation function's extrapolation behavior affects the implied tail densities.
We set the boundaries at $-a=b=4$ by default, implying that roughly $99.99\%$ of the probability mass is modeled flexibly, while the extreme $0.01\%$ relies on the robust parametric extrapolation. Other choices for $a$ and $b$ are possible if required by the application and if enough data are available.

\begin{figure}[t]
    \centering
    \includegraphics[width=\linewidth]{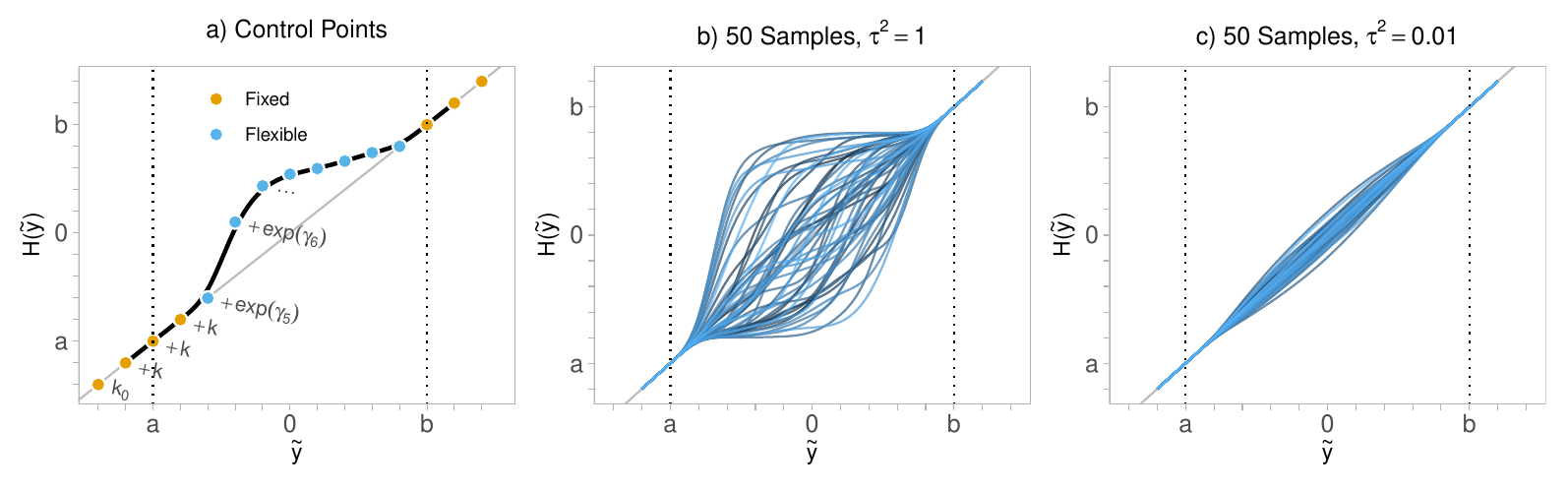}
    \caption{Panel a) illustrates our modified spline using a single random sample of $\bsbeta_i$, which is transformed into $\bsgamma$ using Eq.~(\ref{eq:onion-coefficients}). The points are the B-spline control points $\gamma_{1,i} + \sum_{\ell=2}^j \exp(\gamma_{\ell,i})$, aligned with the knots $k_0, \dots, k_{m+1}$.
        Panels b) and c) each show 50 prior predictive samples of $\calH_i$ obtained by drawing $\bsbeta_i$ from our \textit{onion prior} with variance parameters $\tau^2=1$ (b) and $\tau^2=0.01$ (c).}
    \label{fig:onion-spline}
\end{figure}

\subsection{Regularization for the semi-parametric transformations}
\label{sec:prior-trafo}

Below, while we use the Bayesian terminology of ``prior distributions'' to motivate and construct regularization terms in $\calH$ and induce spatial information pooling in both $\calH$ and $\calG$, note that these priors are employed merely as penalty functions, and parameter estimation is carried out by maximizing the unnormalized log posterior, resulting in maximum a posteriori (MAP) point estimates (see Section~\ref{sec:inference}).

We design a prior for $\bsbeta_i$ with three goals in mind. First, the prior
should regularize estimates of $\calH_i$ towards smoothness to combat overfitting. Second, it
should penalize deviations of the marginal distributions from the base model by regularizing $\calH_i$ towards the identity function. Third, the prior should allow us to borrow information across space to improve estimation in the face of small numbers of training replicates.

\subsubsection{The onion prior for local regularization}

By Theorem~\ref{thm:collapse-to-identity}, the marginal model for $y_i$ collapses to the parametric base model $F_i$, if all elements of $\bsbeta_i$ are equal. We can thus regularize the marginal models towards the base models by penalizing differences in $\bsbeta_i$. We achieve this penalization by placing a first-order random-walk prior on the elements of $\bsbeta_i$, such that $\beta_{i,1} \sim \calN(\beta, \tau^2)$ with starting condition $\beta$, and $\beta_{i,d} \sim \calN(\beta_{i,d-1}, \tau^2)$ for subsequent $d = 2, \dots, D$. (We omit a subscript $i$ on $\tau$ here, as we did not observe improved accuracy using location-specific $\tau$ values in our climate application.) The starting condition $\beta$ is arbitrary, since it cancels out in Eq.~(\ref{eq:onion-coefficients}); we therefore use $\beta=0$ for simplicity.
The resulting joint prior for $\bsbeta_i$ is given by $\bsbeta_i \sim \calN(\bfzero, \tau^2\bfS)$, where $\bfS$ is the (discrete) Brownian motion covariance matrix with entries $\bfS_{[r,c]} = \min(r, c)$.
As $\tau^2 \rightarrow 0$, the elements of $\bsbeta_i$ are forced towards equality, such that by Theorem~\ref{thm:collapse-to-identity}, the marginal model falls back to the base model $F_i$. In general, the elements of $\bsbeta_i$ are regularized towards similarity. This allows us to use a relatively generous default number of parameters (e.g., $D=40$) to capture complex shapes when evident in the data, while automatically reverting to the smooth parametric form elsewhere. This also gives the variance parameter a simple diagnostic interpretation when it is estimated from the data: a fitted value $\hat{\tau}^2$ close to zero indicates that $\calH_i$ is shrunk toward the identity and that the semi-parametric layer contributes little beyond the base model. Panels b) and c) of Figure~\ref{fig:onion-spline} show prior predictive samples of $\calH_i$ based on 50 random draws of $\bsbeta_i$ from this prior for $\tau^2=1$ (Panel b) and $\tau^2 = 0.01$ (Panel c). We find the pattern in Panel b) to be reminiscent of an onion cut in half, hence the name \textit{onion prior}.
This prior setup is closely related to
penalized splines \citep[see][]{Lang2004-BayesianPsplines} and their adaptation in shape-constrained additive models \citep{Pya2015-ShapeConstrainedAdditive}.

\subsubsection{Spatial smoothing}
\label{sec:spatial-smoothing}
We extend the random-walk prior to incorporate spatial dependence using a multi-output GP construction \citep{vanderWilk2020-FrameworkInterdomainMultioutput}. Consider the Cholesky decomposition $\bfS = \bfW\bfW^\top$ of the Brownian motion covariance of $\bsbeta_i$. We write $\bsbeta_i = \bfW \tilde{\bsbeta}_i \in \bbR^{D}$ for $i = 1, \dots, \Ndim$. For each $d = 1, \dots, D$, we define a latent spatial process $\tilde \beta_{d} \sim \mathcal{GP}(0,k_\beta(\cdot, \cdot'))$ with shared amplitude $\tau$ and length scale $\ell^\calH$. The latent coefficient vector $\tilde \bsbeta_i$ is obtained by evaluating these processes at the location indexed by $i$. For $\Ndim$ spatial locations, the resulting joint prior for $\bsbeta = \text{vec}((\bsbeta_1, \dots, \bsbeta_\Ndim)^\sfT)$ collected at all locations has a separable covariance,
$$
\bsbeta \sim \mathcal N\bigl(\bfzero_{\Ndim D},\bfK_{\beta\beta}\bigr), \quad \bfK_{\beta\beta} = \bfS \otimes \widetilde{\bfK}_{\beta\beta},
$$
where the $\Ndim \times \Ndim$ spatial covariance matrix $\widetilde{\bfK}_{\beta\beta}$ is obtained by evaluating $k_\beta(\cdot,\cdot')$ at the $\Ndim$ locations.

\subsection{Maximin ordering of \texorpdfstring{$\bsy$}{y}}
\label{sec:maximin-ordering}
As in \cite{Katzfuss2023-ScalableBayesianTransport}, we use a maximum-minimum-distance (maximin) ordering \citep[e.g.,][]{Guinness2018-PermutationGroupingMethodsb} of the elements of $\bsy$ based on their locations. The ordering is meaningful for the triangular map $\calT$ and provides an effective method for selecting inducing locations for the low-rank approximations of the marginal models, described below in Section~\ref{sec:inference}. Maximin ordering is achieved as follows. The first location is chosen arbitrarily. Subsequent locations are chosen sequentially such that the distance between each new location and its nearest neighbor among the previously chosen locations is as big as possible. Formally, let $\bss_1, \dots, \bss_\Ndim$ denote the unordered locations and $\bss_{i_1}, \dots, \bss_{i_\Ndim}$ the ordered locations, meaning that the ordering is encoded in the indices $i_1, \dots, i_\Ndim$. Arbitrarily start by setting $i_1 = 1$, then for $j = 2, \dots, \Ndim$ choose the next index as $i_j = \argmax_{i \notin \calI_{<j}} \min_{k \in \calI_{<j}} \| \bss_i - \bss_k \|$, where $\calI_{<j} = \{i_1, \dots, i_{j-1}\}$ holds the indices of the already ordered locations. The resulting ordering of $\bsy$ reflects a coarse spatial grid early in the ordering that transitions to finer grids later in the ordering. An illustration for a global grid is given in Figure~\ref{fig:maximin}.

\begin{figure}[bt]
    \centering
    \includegraphics[width=\linewidth]{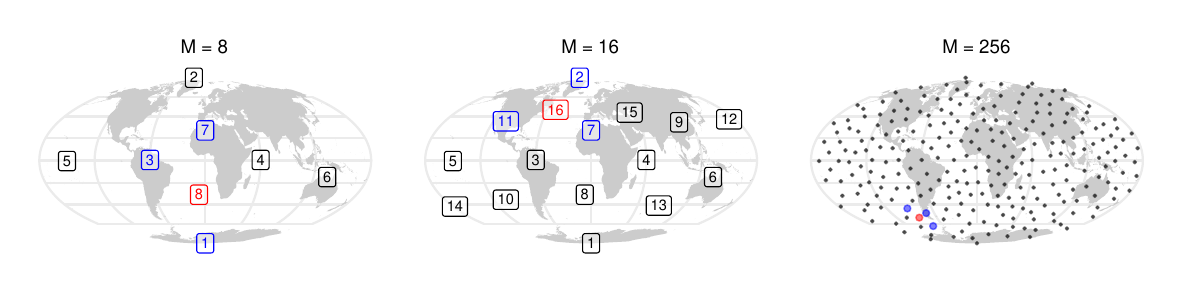}
    \caption{Illustration of the maximin ordering of locations on a global grid, showing the first $M$ locations.
        In the left and middle panels, numbers indicate the position in the ordering. The maximin criterion places each new point to maximize the chordal distance to all preceding points, ensuring uniform global coverage. Apparent asymmetries in the early points arise from projecting this spherical optimization onto a 2D plane for plotting.
        In each panel, the last location in the ordering is displayed in red, and its three nearest neighbors among the preceding locations are displayed in blue.
        }
    \label{fig:maximin}
\end{figure}

\subsection{Review of the scalable Bayesian transport map  \texorpdfstring{$\calT$}{T}}
\label{sec:transport-map}

We adopt the scalable Bayesian transport map of \cite{Katzfuss2023-ScalableBayesianTransport}, which decomposes the complex joint dependence into a sequence of univariate conditional regressions. The components of the triangular transport map $\calT$ are modeled as $\calT_i(\tilde \bsz_{1:i}) = (\tilde z_i - f_i(\tilde \bsz_{<i})) d_i^{-1} \sim \calN(0, 1)$.
Gaussian-process inverse-Gamma priors are placed on $f_i$ and $d_i^2$, with $d_i^2 \sim \calI \calG(\alpha_i, \beta_i)$ and $f_i | d_i^2 \sim \calG \calP(0, d_i^2K_i)$.
The key feature of the setup is the use of structure in $\tilde \bsz$ to drastically reduce the number of parameters that need to be estimated.
Like \cite{Katzfuss2023-ScalableBayesianTransport}, we consider spatial information associated with each element $\tilde z_i$ of $\tilde \bsz$ in the form of locations $\bss_i \in \calS$, $i = 1, \dots, \Ndim$, where $\calS$ is the spatial domain. Now, let $\delta_i = \min_{j \in \{1, \dots, i-1\}} \| \bss_i - \bss_j \|$ denote the minimum pairwise distance between the location $\bss_i$ and any location $\bss_1, \dots \bss_{i-1}$ that precedes $\bss_i$ in the ordering of $\tilde \bsz$.

Intuitively, a small $\delta_i$ means that, among the observations used to predict $\tilde z_i$, there is at least one observation from a nearby location. Thus, the smaller $\delta_i$, the more representative $f_i(\tilde \bsz_{<i})$ is likely to be for $\tilde z_i$, and the smaller the conditional variance $d_i^2$ should be.
\cite{Katzfuss2023-ScalableBayesianTransport} apply results from \cite{Schafer2021-CompressionInversionApproximate} to observe that a roughly power-functional decay of $d_i^2$ with respect to $\delta_i$ is plausible for the considered scenario.
Reflecting this pattern, the prior for $d_i^2$ is set up with $\bbE(d_i^2) = \beta_i / (\alpha_i - 1) = \exp(\theta^d_1 + \exp(\theta^d_2) \log(\delta_i))$ and $\text{SD}(d_i^2) = g \bbE(d_i^2)$, which results in $\alpha_i = 2 + 1 / g^2$ and $\beta_i = \exp(\theta^d_1 + \exp(\theta^d_2) \log (\delta_i))(1 + 1/g^{2})$. The setup ensures that $\bbE(d_i^2)$ is increasing in $\delta_i$.
The constant $g \in \bbR_{>0}$ scales the prior standard deviation of $d_i^2$ relative to its expectation, where the default $g=4$ is used to obtain a relatively weak prior.

For the first location, we use $f_1(\tilde{\bsz}_{<1}) = 0$ and $\delta_1 = \delta_2^2 / \delta_5$. The latter choice is motivated by the observation that the minimum pairwise distance decays roughly as $\delta_i\propto i^{-1/\mathrm{dim}}$ \citep{Katzfuss2023-ScalableBayesianTransport}. We normalize all pairwise minimum distances by dividing by $\delta_1$, so that $\exp(\theta_1^d)$ is the expected variance at location $1$, while $\theta^d_2$ controls how
the expected variance scales with increasing distance $\delta_i$.

The covariance function $K_i$ in the Gaussian-process prior for $f_i$ is designed to shrink $f_i$ towards linearity with decreasing $\delta_i$ and let the contribution of each input $\tilde z_1, \dots, \tilde z_{i-1}$ decay exponentially according to its position when ordered by increasing distance to the location of $\tilde z_i$. These goals are achieved by letting
\begin{equation}
    \label{eq:tm-covariance}
    K_i(\tilde \bsz_{o(i)}, \tilde \bsz_{o(i)}') = \bbE(d_i^2)^{-1}
    \left[
        \tilde \bsz^\sfT_{o(i)}\bsQ_i \tilde \bsz'_{o(i)}
        + \sigma^2_i
        \rho \left(
        \gamma^{-1} \sqrt{
            (\tilde \bsz_{o(i)} - \tilde \bsz'_{o(i)})^\sfT\bsQ_i (\tilde \bsz_{o(i)} - \tilde \bsz'_{o(i)})
        }
        \right)
        \right]
\end{equation}
for $i = 2, \dots, \Ndim$,
which requires some elaboration. First, $\tilde \bsz_{o(i)}$ denotes a view of $\tilde \bsz_{<i}$ that is ordered in increasing distance to the location of $\tilde z_i$, such that $\tilde z_{o(i)_1}$ is the element of $\tilde \bsz_{<i}$ that is closest to $\tilde z_i$. Second, $\bsQ_i$ is a diagonal matrix with the diagonal element indexed by $j$ given by $\exp(-j\exp(\theta^q))$ for $j = 1, \dots, i-1$.
This way, $\bsQ_i$ both controls the linear part of $f_i$ and encodes the decreasing contributions of inputs to $f_i$ according to their distance to the location of $\tilde z_i$. Due to the ordering, larger $j$ imply higher distance to the location of $\tilde z_i$, and thus the diagonal entries of $\bsQ_i$ are decreasing in $j$.
The hyperparameter $\theta^q$ controls how quickly they approach zero; with larger $\theta^q$ implying a steeper decrease.
Third, $\rho$ denotes a covariance function, where $\gamma = \exp(\theta^\gamma)$ serves as a range parameter, such that $\theta^\gamma$ is the log range. We use the Matérn covariance function with smoothness $3/2$. Fourth, $\sigma_i^2 = \exp(\theta^\sigma_1 + \exp(\theta^\sigma_2) \log (\delta_i))$ determines the extent of nonlinearity in $f_i$ and decays similarly to $d_i^2$, rendering $f_i$ to resemble a linear function more closely as $\delta_i$ decreases. The meaning of the hyperparameters $\theta^\sigma_1$ and $\theta_2^\sigma$ is analogous to $\theta_1^d$ and $\theta_2^d$, respectively.

The regression functions $f_i$ can be interpreted according to their position in the maximin ordering.
Consider Figure~\ref{fig:maximin}: for small $i$, the locations preceding $\bss_i$ are spatially dispersed, so that even the nearest previously ordered inputs to $f_i$ are relatively far from $\bss_i$ and the corresponding regressions primarily capture long-range dependence.
As $i$ increases, the nearest preceding locations become progressively closer to $\bss_i$, and the regressions increasingly capture fine-scale, local dependence.
The model thus achieves a smooth transition from coarse to fine spatial scales.

Notably, the model does not require explicit information on geographical features such as topography or whether a point is located on land or water. Instead, their influence is represented indirectly, through the empirical dependence patterns captured by the regressions $f_i$, since neighboring locations need not have the same predictive relevance if their observed relationship with $\tilde z_i$ differs.

In sum, the scalable Bayesian transport map is governed by the six hyperparameters $\bstheta^\calT = (\theta^d_1, \theta^d_2, \theta^q, \theta^\gamma, \theta^\sigma_1, \theta^\sigma_2)^\sfT \in \bbR^6$. For given $\tilde z_i$ and $\bstheta^\calT$, the Gaussian-process inverse-gamma prior on $(f_i, d_i^2)$ is conjugate, such that a posterior transport map is available in closed form.

\section{Parameter estimation and computational complexity}
\label{sec:inference}

The model involves three sets of unknown parameters: the parametric marginal parameters
$\bszeta$ and hyperparameters $\bstheta^\calG$ in $\calG$, where $\bszeta = \operatorname{vec}((\bszeta_1, \dots, \bszeta_\Ndim)^\sfT)$, the semi-parametric marginal parameters $\bsbeta$ and hyperparameters $\bstheta^\calH$ in $\calH$, where $\bsbeta = \operatorname{vec}((\bsbeta_1, \dots, \bsbeta_\Ndim)^\sfT)$, and the transport-map hyperparameters $\bstheta^\calT$ in $\calT$.
Joint estimation of all parameters is computationally intractable due to the combination of the non-Gaussian semi-parametric likelihood and the massive dimensionality. Therefore, we adopt a two-stage estimation strategy, common in copula literature \citep{Joe2005-AsymptoticEfficiencyTwostage}: we first estimate the marginal models (parametric + semi-parametric) jointly under a working assumption of spatial independence of $\tilde \bsz$, and then estimate the spatial dependence in $\tilde \bsz$ (transport map) given the fitted marginals.

\subsection{Stage 1: Scalable marginal-model estimation}
\label{sec:stage1}

\subsubsection{Low-rank approximation}
\label{sec:stage1-lowrank}
Direct evaluation of the priors for $\bszeta$ and $\bsbeta$ scales as $\calO(P\Ndim^3)$ and $\calO(\Ndim^3)$, respectively, due to the necessity to invert the $\Ndim \times \Ndim$ spatial covariance matrices $\bfK_{\zeta_p \zeta_p}$ and $\widetilde \bfK_{\beta \beta}$.
We overcome this bottleneck using a low-rank approximation based on $M \ll \Ndim$ inducing locations, selected as the first $M$ points in a maximum-minimum-distance ordering (see Section~\ref{sec:maximin-ordering} and Figure~\ref{fig:maximin}), which we describe for $\bsbeta$ and apply analogously to~$\bszeta$.
We parameterize the multivariate spatial field $\bsbeta \in \mathbb{R}^{\Ndim D}$ through whitened inducing variables $\bfu \sim \calN(\bfzero_{MD}, \bfI_{MD})$ at these $M$ locations, writing
\begin{equation}
    \label{eq:beta-basis}
    \bsbeta = \bfB(\bstheta^\calH)\bfu, \quad \text{where} \quad \bfB(\bstheta^\calH) = \bfK_{\beta u} \bfK_{uu}^{-1}\bfL,
\end{equation}
with $\bfK_{uu} = (\bfS \otimes \widetilde \bfK_{uu})$ and $\bfL\bfL^\sfT = \bfK_{uu}$, such that $\bfL$ is the lower Cholesky factor of the covariance matrix at the locations of the inducing variables, $\bfL \bfu \sim \calN(\bfzero_{MD}, \bfK_{uu})$.
This approach corresponds to projected process prediction in \cite{Rasmussen2008-GaussianProcessesMachine}, with implied prior covariance $\bfK_{\beta\beta} = \bfK_{\beta u} \bfK_{uu}^{-1} \bfK_{u\beta}$, and can be interpreted as a basis-function parameterization with the covariance function acting as a (hyperparameter-dependent) basis function and $\bfu$ acting as the coefficient vector. Instantiation of the full $(D\Ndim \times DM)$ matrix $\bfK_{\beta u}$ can be avoided by exploiting its Kronecker structure to rewrite the model as
$\bsbeta = \operatorname{vec}(\widetilde \bfK_{\beta u}\widetilde \bfL^{-\sfT} \bfU \bfW^\sfT)$, where $\widetilde\bfL \widetilde\bfL^\sfT = \widetilde\bfK_{uu}$, $\bfW \bfW^\sfT = \bfS$, $\bfU$ is $M\times D$ so that $\bfu = \operatorname{vec}(\bfU)$, and $\widetilde \bfK_{\beta u}$ of dimension $(\Ndim \times M)$ is the largest matrix. This low-rank approximation reduces the dominant computational cost to the Cholesky factorization of the $(M \times M)$ matrix $\widetilde \bfK_{uu}$, which scales as $\calO(M^3)$, and a factor $\calO(\Ndim MD)$ from matrix multiplication. Consequently, the prior evaluation cost becomes linear in $\Ndim$.
For $\bszeta$, we analogously have $P$ independent expressions $\bszeta_{\cdot, p} = (\zeta_{1,p}, \dots, \zeta_{\Ndim,p})^\top = \bfB(\bstheta^\calG_p)\bfv_p$, with $\bfv_p \sim \calN(\bfzero_{M}, \bfI_{M})$ and $\bszeta = \operatorname{vec}((\bszeta_{\cdot,1}, \dots, \bszeta_{\cdot,P}))$.
The inducing-point representation also acts as a regularizer for the spatial fields of marginal parameters and spline coefficients. This borrowing of information is helpful when $\Nobs$ is small, but, like any low-rank spatial smoothing technique, it may attenuate abrupt local changes in marginal behavior. If the true marginal distributions exhibit more local variation than the low-rank model can capture, we would expect spatially coherent patterns of marginal lack of fit. Such patterns can be assessed via marginal-only comparisons and diagnostic maps of quantile discrepancies, as discussed in Section~\ref{sec:global-data} and the Appendix \citep{Brachem2026-SCT-Appendix}.

\subsubsection{Estimation}
\label{sec:stage1-estimation}
Parameter estimation in Stage 1 aims to achieve $\tilde z_i = (\calH_i^{\bsbeta_{i}} \circ \calG_i^{\bszeta_i})(y_i) \sim \calN(0, 1)$, $i = 1, \dots, \Ndim$,
where we write $\calH_i^{\bsbeta_i}$ to explicitly denote the dependence of $\calH_i$ on the parameters $\bsbeta_i$.
To ease notation, we do not explicitly denote the dependence of $\bsbeta_i$ and $\bsbeta$ on $\bfu$ and $\bstheta^\calH$ (cf.\ Eq.~(\ref{eq:beta-basis})).
The same convention applies to $\calG_i$ and $\bszeta_i$, with $\bszeta_i$ depending on $\bfv$ and $\bstheta^\calG$.
We employ a working independence assumption (which is dropped in Stage 2), constructing an auxiliary joint density
\begin{equation}
    p(\bsy | \bfv, \bfu) = \prod_{i=1}^\Ndim
    \left(
    \bigl(\phi \circ \calH_i^{\bsbeta_i} \circ \calG_i^{\bszeta_i}\bigr)(y_i)
    \times \left| \frac{\partial\, \calH_i^{\bsbeta_i}(\tilde y_i)}{\partial \tilde y_i} \right|
    \times \left| \frac{\partial\, \calG_i^{\bszeta_i}(y_i)}{\partial y_i} \right|
    \right),
\end{equation}
where $\phi$ denotes the standard Gaussian density and $\tilde y_i = \calG_i(y_i | \bszeta_i)$. Given $\Nobs$ independent replicates $\bsy_j$, $j = 1, \dots, \Nobs$, we apply gradient ascent to obtain joint maximum a posteriori (MAP) point estimates of the parameters and hyperparameters:
\begin{equation}
    \bigl(\hat \bfv, \hat \bfu, \hat \bstheta^\calG,  \hat \bstheta^\calH\bigr)
    = \argmax_{\bfv, \bfu, \bstheta^\calG, \bstheta^\calH}
    \left(
    \prod_{j=1}^{\Nobs} p\bigl(\bsy_j | \bszeta, \bsbeta\bigr)
    \right)
    p(\bfu)
    \prod_{p=1}^P
    p(\bfv_p)\,
    .
\end{equation}
For most covariance functions, each $\bstheta_{p}^\calG$, $p=1, \dots, P$, and $\bstheta^\calH$ will consist of a length scale $\ell$ and an amplitude $\tau$ (omitting any indices), both of which are restricted to the positive real line. We therefore use softplus transformations $\ell = \log(1 + \exp(\eta_\ell))$ and $\tau^2 = \{\log(1 + \exp(\eta_\tau))\}^2$ and optimize over the real-valued parameters $\eta_\ell$ and $\eta_\tau$ using constant (uniform) priors. On the levels of $\ell$ and $\tau^2$, this implies improper priors that concentrate probability mass near zero but have a heavy right tail. See Section~B of the Appendix \citep{Brachem2026-SCT-Appendix} for more details.

\subsection{Stage 2: Estimation of the transport map \texorpdfstring{$\calT$}{T}}
\label{sec:transport-map-inference}

In the second stage, we fix the marginal transformations using the Stage 1 estimates to obtain pseudo-data $\hat{\tilde \bsz} = (\calH^{\hat \bsbeta} \circ \calG^{\hat \bszeta})(\bsy)$, where $\hat{\bsbeta} = \bfB(\hat{\bstheta}^\calH) \hat{\bfu}$ and $\hat{\bszeta}_p = \bfB(\hat{\bstheta}_p^\calG) \hat{\bfv}_p$ for all $p = 1, \dots, P$. Estimation aims to achieve $\calT(\hat{\tilde \bsz}) \sim \calN(\bfzero, \bfI)$, such that the second-stage joint density is
\begin{equation}
    \label{eq:density-T}
    p(\hat{\tilde \bsz}) =  \prod_{i=1}^\Ndim
    \left(
    \phi\bigl( \calT_i(\hat {\tilde \bsz}_{1:i}) \bigr) \times \left|
    \frac{\partial\, \calT_i(\hat {\tilde \bsz}_{1:i})}{\partial \hat{\tilde z}_i}
    \right|
    \right).
\end{equation}
This allows us to directly apply the empirical Bayes estimation approach developed by \cite{Katzfuss2023-ScalableBayesianTransport}, which we briefly summarize here.
Through $\calT_i$, the joint density Eq.~(\ref{eq:density-T}) depends on $f_i$ and $d_i^2$; $i = 1, \dots, \Ndim$, which in turn depend on the hyperparameters $\bstheta^\calT$ as described in Section~\ref{sec:transport-map}.
We collect the sample of
$\Nobs$ pre-transformed observations $\hat{\tilde{\bsz}}_j$
in a
$\Nobs \times \Ndim$ matrix $\hat{\tilde \bfZ} = (\hat{\tilde{\bsz}}_1, \dots, \hat{\tilde{\bsz}}_\Nobs)^\sfT$, with row $j$ holding realization $\hat{ \tilde \bsz}_j$ for all $j=1, \dots, \Nobs$.
Proposition 2 in \cite{Katzfuss2023-ScalableBayesianTransport} gives an expression for the marginal likelihood $p(\hat{\tilde \bfZ}| \bstheta^\calT)$ with the $f_i$ and $d_i^2$ integrated out.
We maximize this marginal likelihood using gradient-based optimization to find $\hat \bstheta^\calT = \argmax_{\bstheta^\calT} p(\hat{\tilde \bfZ} | \bstheta^\calT)$.
Given $\bstheta^\calT$, we obtain the elements of the posterior transport map by Proposition 1 in \cite{Katzfuss2023-ScalableBayesianTransport} in closed form.
The scalability of estimation for $\calT$ is improved by a Vecchia-type approximation \citep{Kidd2022-BayesianNonstationaryNonparametric,Katzfuss2021-GeneralFrameworkVecchia}, making the simplifying assumption that $f_i$ depends on a subset of the elements of $\tilde \bsz_{<i}$ instead of the full vector: $f_i(\tilde \bsz_{<i}) = f_i(\tilde \bsz_{c(i)})$ for some $c(i) \subseteq \{1, \dots, i-1\}$. The size of the conditioning set $c(i)$ is based on the relevance decay in the prior for $f_i$, such that only relevant elements of $\tilde \bsz_{<i}$ are included. Specifically, the maximum size of $c(i)$ is set to $m = \max\{j \geq 1 : \exp(-j \exp(\theta^q)) \geq \epsilon \}$, with $\epsilon = 0.01$. This reduces the computational complexity for determining the $i$th posterior-map element from $\calO(\Nobs^3 + i\Nobs^2)$ to at most $\calO(\Nobs^3 + m\Nobs^2)$, with $m$ typically $<10$.

\subsection{Scope of uncertainty quantification} \label{sec:uncertainty}

Our estimation strategy accounts for uncertainty only partially. The marginal transformations are fixed at regularized MAP estimates, and the transport-map hyperparameters are fixed at empirical-Bayes estimates. Such a plug-in treatment of marginal models is common in copula applications, where marginals are often estimated before fitting the dependence structure \citep{Joe2005-AsymptoticEfficiencyTwostage}. Conditional on these point estimates, the posterior predictive transport map accounts for uncertainty in the conjugate quantities $f_i$ and $d_i^2$, but it does not propagate uncertainty about the marginal transformations into the dependence model.

This limitation is most relevant for summaries that depend on marginal tail behavior, especially when $\Nobs$ is very small or relevant tail regions are weakly represented. In our application, this includes threshold-based severe-rainfall and dry-condition probability maps: small changes in the fitted marginals can affect local exceedance probabilities and, through the transformed pseudo-data, estimates of spatial co-occurrence. Predictive-sample intervals may therefore be too narrow in such settings.

A full-model bootstrap could propagate this uncertainty in principle, but would require many complete global refits. A more principled route is joint inference over the marginal transformations and the dependence model, possibly through structured variational approximations that retain much of the present sparsity.

\subsection{Computational complexity}

Assuming $\Ndim > M > \Nobs > D > m > P$, the overall time complexity of a single evaluation of the unnormalized log posterior in our approach can be summarized as $\mathcal{O}(\Ndim D(M + \Nobs) + M^3P)$ for our marginal model and $\mathcal{O}(\Ndim\Nobs^3)$ for the Bayesian transport-map dependence model. Inverse evaluations scale linearly in $\Nobs\Ndim$ except for the transport map inversion, which introduces a quadratic dependence on $\Nobs$ per recursion step. For optimization, gradients are obtained via automatic differentiation in JAX \citep{deepmind2020jax}, so gradient evaluations scale like log-posterior evaluations up to a constant-factor overhead.
In typical applications, the overall runtime is primarily driven by operations linear in $\Nobs\Ndim$ and cubic in $\Nobs$ (which is typically small).
A detailed summary of all computational costs is provided in Section~C of the Appendix \citep{Brachem2026-SCT-Appendix}.
We discuss possible use of the fitted model for compressed storage, including its relation to storage-oriented stochastic generators and a conditional-completion illustration, in Section~D of the Appendix \citep{Brachem2026-SCT-Appendix}.

\section{Reproducing internal climate variability}
\label{sec:application}

\subsection{Climate data and comparison set-up}
\label{sec:application-setup}

We demonstrate our method's capacity to generate high-dimensional climate fields using global log-precipitation-rate fields
from the Community Earth System Model (CESM) Large Ensemble Project \citep{Kay2015-CommunityEarthSystem}, as mentioned in the introduction.
While CESM simulates the full climate system, generating large ensembles of a single target field, such as precipitation rate, still requires the full computational cost of the underlying model, on the order of $10^5$ CPU hours per ensemble member. Thus, a data-efficient stochastic surrogate offers immense scientific value.
Recall that the data reside on a roughly 1\degree\ longitude-latitude global grid of size $\Ndim = 288 \times 192 = 55{,}296$. We use 98 samples, each representing a July~1 daily-average field from one of 98 consecutive CESM model years starting in the year 402. The grid is trimmed at the poles to remove duplicate locations, yielding a field of size $\Ndim = (288 \times 190) + 2 = 54{,}722$ for practical use. To avoid numerical problems with log-transforming exact zeros in the raw data, we add a small constant $\epsilon = 10^{-10}$ before applying the transformation.

\begin{figure}[bt]
    \centering
    \includegraphics[width=\linewidth]{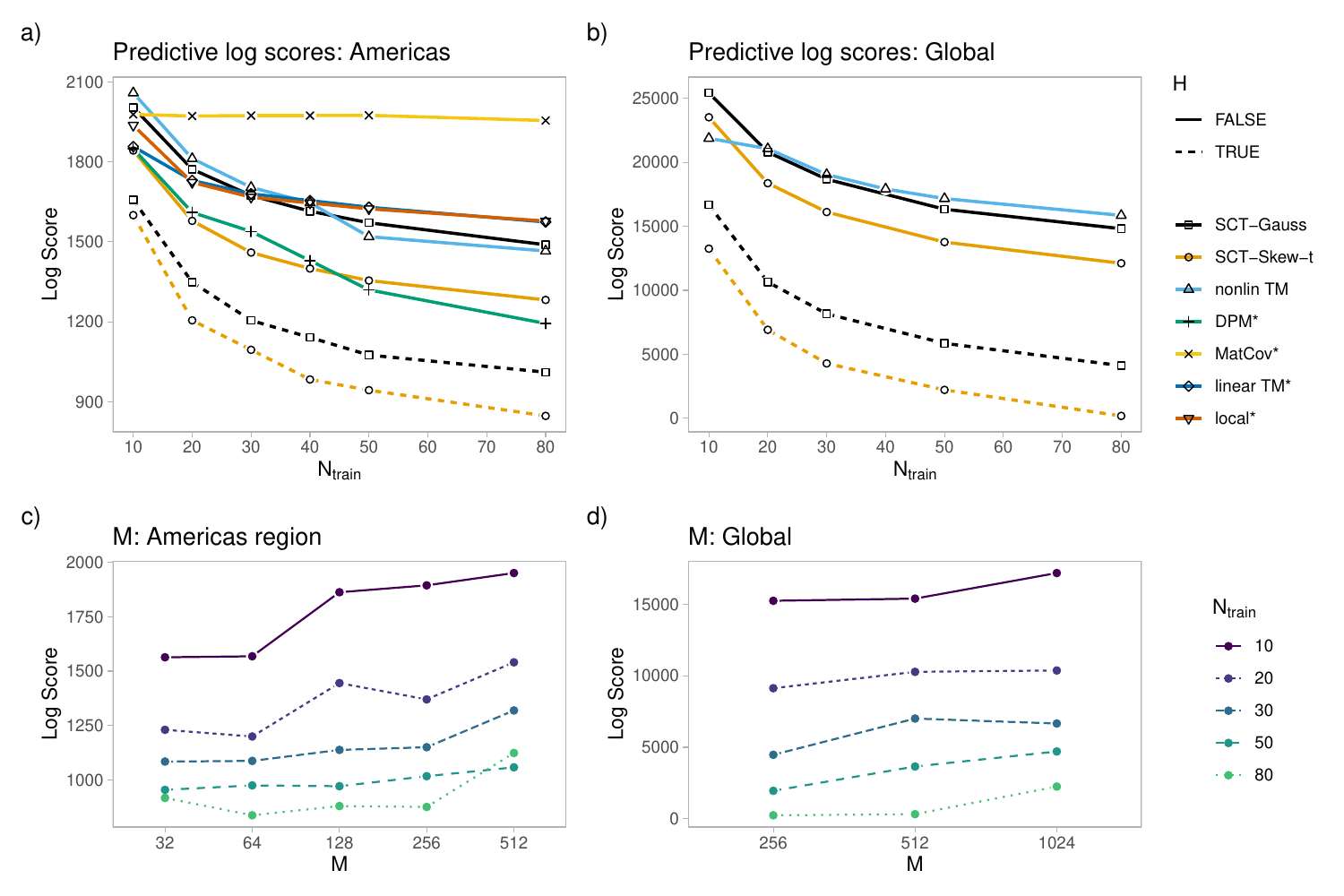}
    \caption{Our SCT models yield substantially more accurate predictive distributions than existing methods, as measured by log scores (lower is better). The top row shows average predictive log scores over five random train/test splits.
        Dashed lines indicate models that include a fitted semi-parametric correction layer $\calH$.
        Results marked with an asterisk (*) are taken from \cite{Katzfuss2023-ScalableBayesianTransport}, Fig.~10, without re-running the models.
        The bottom row shows predictive log scores in the \texttt{SCT-Skew-t} specification (including a fitted semi-parametric layer $\calH$) as a function of the number of inducing points $M$ and training samples $\Ntrain$. \label{fig:logscores}}
\end{figure}

We evaluate several variants of our scalable composite transformations (SCT) model against state-of-the-art baselines:
\begin{itemize}
    \item \texttt{nonlin TM}: Fixes both marginal transformation layers to the identity map, yielding the nonlinear Bayesian transport map as proposed by \cite{Katzfuss2023-ScalableBayesianTransport}.
    \item \texttt{SCT-Gauss}: $\calG_i$ is defined with components $\calG_i(y_i) = \Phi^{-1}\bigl( \Phi(y_i | \mu_i, \sigma_i^2) \bigr) = (y_i - \mu_i)/\sigma_i$, i.e. using a Gaussian distribution with mean and variance varying over locations. We estimate the inverse softplus transformations of the standard deviations, $\eta_i = \log(\exp(\sigma_i) - 1)$, such that we have $\bszeta_i = (\mu_i, \eta_i)$.
    \item \texttt{SCT-Skew-t}: $\calG_i$ is defined with components $\calG_i(y_i) = \Phi^{-1}\bigl( F(y_i | \mu_i, \sigma_i, \alpha_i, \nu) \bigr)$, where $F$ is the CDF of a Type-3 skewed \textit{t}-distribution \citep[see][]{Rigby2020-DistributionsModelingLocation} with location $\mu_i \in \bbR$, scale $\sigma_i \in \bbR_{>0}$, skewness $\alpha_i \in \bbR_{>0}$, and degrees of freedom $\nu\in \bbR_{>0}$. With $\alpha_i=1$, the symmetric \textit{t}-distribution is retrieved; with $\alpha_i > 1$ and $0 < \alpha_i < 1$, the distribution becomes right- and left-skewed, respectively. We estimate the positive-valued parameters on inverse-softplus level, $\eta_{\cdot} = \log(\exp(\cdot) - 1)$, such that we have $\bszeta_i = (\mu_i, \eta_{\sigma_i}, \eta_{\alpha_i}, \eta_{\nu})^\sfT$.
          The degrees of freedom $\nu$ do not vary across locations.
\end{itemize}
We assess each configuration with and without the semi-parametric correction layer $\calH$. When active, $\calH$ uses a global random-walk variance $\tau^2$, boundaries $-a=b=4$, and $D=40$ flexible parameters. In all models, the transport map $\calT$ is defined as in Section~\ref{sec:transport-map}.
These model variants also help interpret the roles of $\calG$ and $\calH$. Because the two marginal layers are estimated jointly, their effects cannot generally be decomposed uniquely from a single fitted model: $\calH$ can absorb structure that could also be represented by the base family in $\calG$. The paired variants therefore serve as empirical contrasts: switching $\calH$ on or off shows the added value of the semi-parametric correction for a fixed base family, while changing $\calG$ under the same $\calH$ setting shows how much the choice of base family still matters.

We evaluate the models in terms of their predictive performance using log scores,
defined as the negative expected log predictive density
$-\mathbb{E}_{\bsy^*}[\log p(\bsy^* \mid \bfY, \hat{\bstheta})]$,
for a new test field $\bsy^*$. Evaluation is conditional on the point estimates
$\hat{\bstheta}=(\hat{\bszeta},\hat{\bsbeta},\hat{\bstheta}^{\calT})$, while integrating over the
conditionally conjugate posterior distributions of $f_i$ and $d_i^2$.
Thus, when evaluating the joint density in \eqref{eq:joint-density} at a
held-out field $\bsy^*$, we use $\hat{\bszeta}$ in $\calG$,
$\hat{\bsbeta}$ in $\calH$, and replace $\calT$ by the posterior predictive
transport map $\widetilde{\calT}$ obtained under
$\hat{\bstheta}^{\calT}$ \citep[Proposition~1]{Katzfuss2023-ScalableBayesianTransport}.
The log score is equivalent, up to an additive constant, to the
Kullback-Leibler divergence between the unknown true distribution and the
model approximation. We approximate the expectation by averaging over 14
holdout replications, which are not used for either estimation or early
stopping, and over five random train/test splits.

\subsection{Americas subregion}
\label{sec:americas}

First, we consider a subset of the full field of size $\Ndim = 37 \times 74 = 2{,}738$ in a region containing parts of the Americas ($30$\degree S to $40$\degree N and $115$\degree W to $70$\degree W). This subset allows for performance comparisons to alternative approaches that do not scale well to the full global dataset. These include \texttt{local}, a locally parametric method for climate data that fits anisotropic Matérn covariances in local windows and combines the local fits into a global model \citep{Wiens2020-ModelingSpatialData}, \texttt{MatCov}, a Gaussian model with zero mean and isotropic Matérn covariance, \texttt{linear TM}, a simplified version of $\calT$ that sets all map elements to linear functions, and \texttt{DPM}, a version of $\calT$ that incorporates a Dirichlet process mixture model to allow non-Gaussian distributions for the errors $\epsilon_i$ in the nonlinear regressions $\tilde z_i = f_i(\tilde \bsz_{<i}) + \epsilon_i$ of the map elements. In \cite{Katzfuss2023-ScalableBayesianTransport}, \texttt{DPM} showed the best predictive performance on the Americas data, but since it requires MCMC sampling, it does not scale well to large spatial fields.

\begin{figure}[tb]
    \centering
    \includegraphics[width=\linewidth]{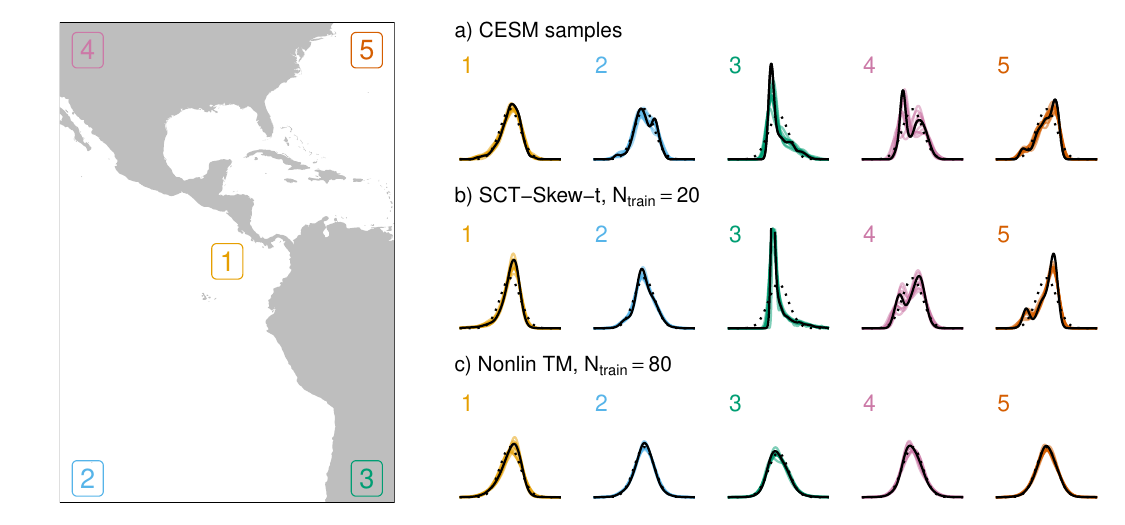}
    \caption{Marginal distributions at and near five locations in the Americas region. The selected locations are the first five locations of a maximin ordering of all locations in the region. Each numbered panel shows data for the respective location and the nine closest locations. The solid black lines show kernel density estimates obtained from pooling the data from all ten locations. The colored lines show kernel density estimates of the data at the individual locations. The dotted line is a standard Gaussian density for reference.
        Row a) uses all 98 replicates in the data, row b) uses 1000 samples from \texttt{SCT-Skew-t} trained on $\Ntrain=20$ replicates, and row c) uses 1000 samples from \texttt{nonlin TM} trained on $\Ntrain = 80$ replicates.
        \label{fig:americas_ip}}
\end{figure}

Our results for the Americas subregion are shown in panel a) of Figure~\ref{fig:logscores}. We use $M = 64$ inducing locations for the low-rank approximations in the marginal models.
For all considered composite models, estimating the
semi-parametric correction dramatically improves predictive performance even with small sample sizes. The best overall performance is obtained by the model with semi-parametric correction and skew-$t$ marginal. All of our full composite transformation models outperform \texttt{DPM} (and all other competitors). Remarkably, \texttt{SCT-Skew-t} achieves a better log score with $\Ntrain=20$ than the previous state-of-the-art (\texttt{DPM}) achieves with $\Ntrain = 80$.
The comparison between \texttt{SCT-Gauss} and \texttt{SCT-Skew-t} is also informative about the role of the parametric base family: even when both specifications include $\calH$, the skew-$t$ base model outperforms the Gaussian base model, indicating that, while $\calH$ mitigates the importance of choosing a suitable parametric baseline, the specific choice still matters in practice.
The scores taken from \cite{Katzfuss2023-ScalableBayesianTransport} are shifted to correct for different preprocessing, see Section~E of the Appendix \citep{Brachem2026-SCT-Appendix} for details.

Figure~\ref{fig:americas_ip} contrasts the standardized marginal distributions of the CESM data against samples from \texttt{SCT-Skew-t} ($\Ntrain = 20$) and \texttt{nonlin TM} ($\Ntrain = 80$) at five areas in the Americas. The data (row a) exhibit similar distributional features at neighboring locations but significant spatial heterogeneity across the five areas, ranging from near-Gaussian (Location 1) to pronounced skewness (Location 3) and multimodality (Location 4).
Despite using only $\Ntrain=20$ replicates, \texttt{SCT-Skew-t} (row b) faithfully reproduces these complex features, capturing the heavy tails and bimodal structures that characterize the local precipitation regimes. Conversely, the \texttt{nonlin TM} (row c) reverts to a Gaussian-like profile even with $\Ntrain = 80$, effectively smoothing out the local asymmetries and extremes essential for accurate risk quantification.

\subsection{Global data}
\label{sec:global-data}

\begin{figure}
    \centering
    \includegraphics[width=\linewidth]{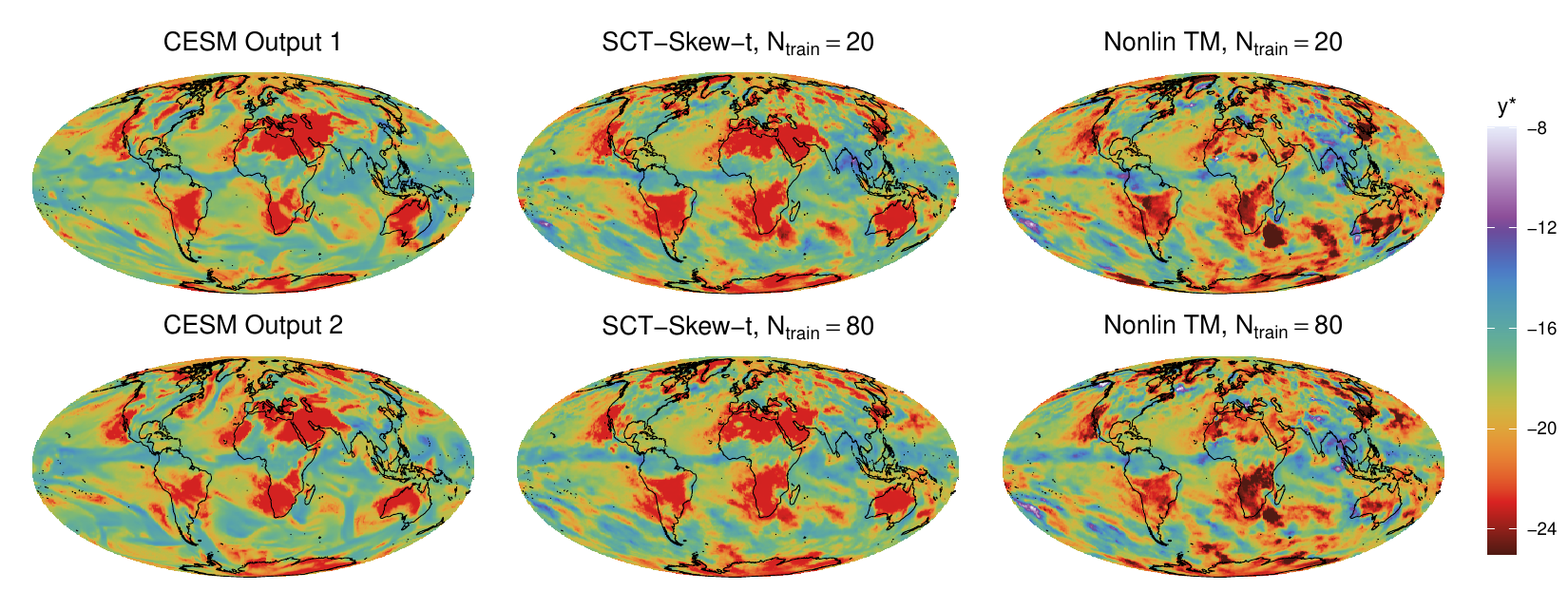}
    \caption{Comparison of climate-model output (left column) to samples from \texttt{SCT-Skew-t} including $\calH$ (middle column) and \texttt{nonlin TM} (right column). Ideally, the surrogate samples should look like they are drawn from the same distribution as the climate-model output. To facilitate comparisons between the samples from the surrogate models, all samples in the right four panels are generated by inverting the transformations on the same random noise realization $\bsz^*$ drawn from $\calN(\bfzero, \bfI)$. To keep the color scale in a reasonable range, we truncate the data to $[-25, -8]$, which is relevant mainly for the samples from \texttt{nonlin TM}.}
    \label{fig:samples-precip-daily-avg}

\end{figure}

For the global dataset, we observe the same pattern in log scores as for the Americas, see Figure~\ref{fig:logscores} b). We use $M=256$ for the SCT models. Since DPM is not viable for the full global dataset, we consider \texttt{nonlin TM} as the state-of-the-art benchmark. Most notably, \texttt{SCT-Skew-t} substantially outperforms \texttt{nonlin TM}, yielding a lower log score using $\Ntrain=10$ than \texttt{nonlin TM} achieves using $\Ntrain = 80$. By reducing the training requirement from 80 to 10~runs, our method effectively provides an eight-fold increase in the computational budget for climate research, allowing scientists to explore eight times as many emission scenarios for the same cost.

In Figure~\ref{fig:samples-precip-daily-avg}, we compare samples for the global grid generated from the physical climate model (left column, labeled \textit{CESM Output 1} and \textit{2}) to samples from \texttt{SCT-Skew-t} and \texttt{nonlin TM}, trained on $\Ntrain = 20$ and $\Ntrain=80$ field samples. Consistent with the quantitative log-score results, the \texttt{SCT-Skew-t} model produces samples that look qualitatively similar to the CESM output, reproducing both the large-scale and small-scale spatial structure, even when trained on only 20 examples. Increasing $\Ntrain$ to $80$ refines the representation without substantially altering the overall appearance.
\texttt{Nonlin TM} also reproduces the main large-scale features of the CESM data, but its samples exhibit noticeably stronger contrasts and more pronounced regional patterns than those present in the CESM output. This becomes particularly apparent when contrasting \texttt{SCT-Skew-t} at $\Ntrain=20$ with \texttt{nonlin TM} at
$\Ntrain=80$: despite having considerably more training data, the \texttt{nonlin TM} sample still exhibits amplified variation, while \texttt{SCT-Skew-t} already provides a close visual match to the CESM output's distribution with far fewer training cases.

\subsubsection{Marginal plausibility}
Figure~\ref{fig:PYsmall} compares the spatial distribution of severe rainfall (Panel a) and dry conditions (Panel b) across the CESM samples and the two surrogate models. In both cases, \texttt{SCT-Skew-t} trained on only $\Ntrain = 20$ fields (middle column) reproduces the spatial structure visible in the CESM data better than \texttt{nonlin TM} trained on $\Ntrain = 80$. The probability of severe rainfall is overestimated by \texttt{nonlin TM} across large areas including North America, Europe, Russia, and equatorial Africa. In contrast, \texttt{SCT-Skew-t} shows less pronounced overestimation, which is further reduced with increased training sample size (not shown).
Considering dry conditions, \texttt{nonlin TM} underestimates the probability of low precipitation, while \texttt{SCT-Skew-t} closely reproduces both the spatial pattern and the relative magnitude of low-precipitation probabilities.
These results corroborate the quantitative results in terms of improved log-scores, suggesting that summaries drawn from \texttt{SCT-Skew-t} are more consistent than those from \texttt{nonlin TM} with the observed spatial distribution of wet and dry conditions in the CESM data, even in small-data settings.

\begin{figure}
\centering
    \includegraphics[width=\linewidth]{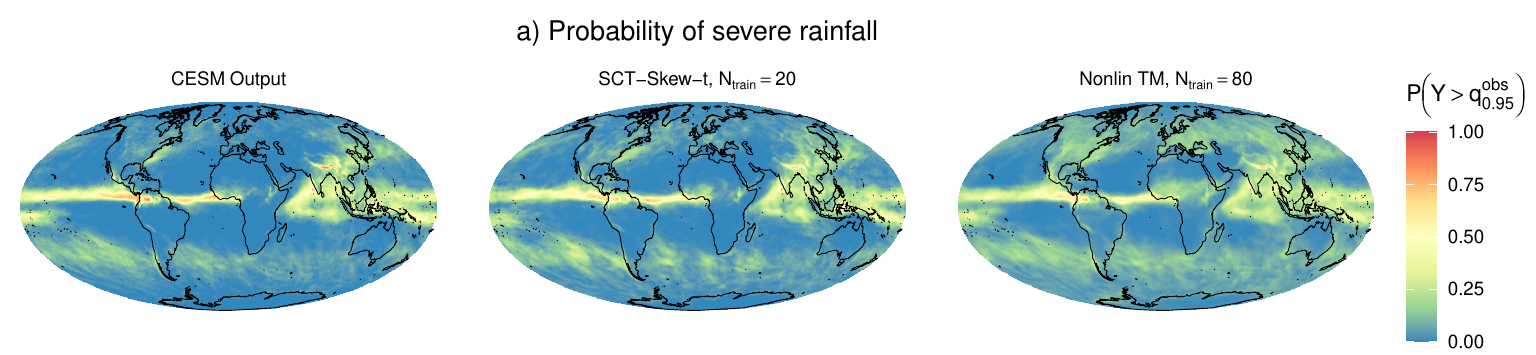}
    \includegraphics[width=\linewidth]{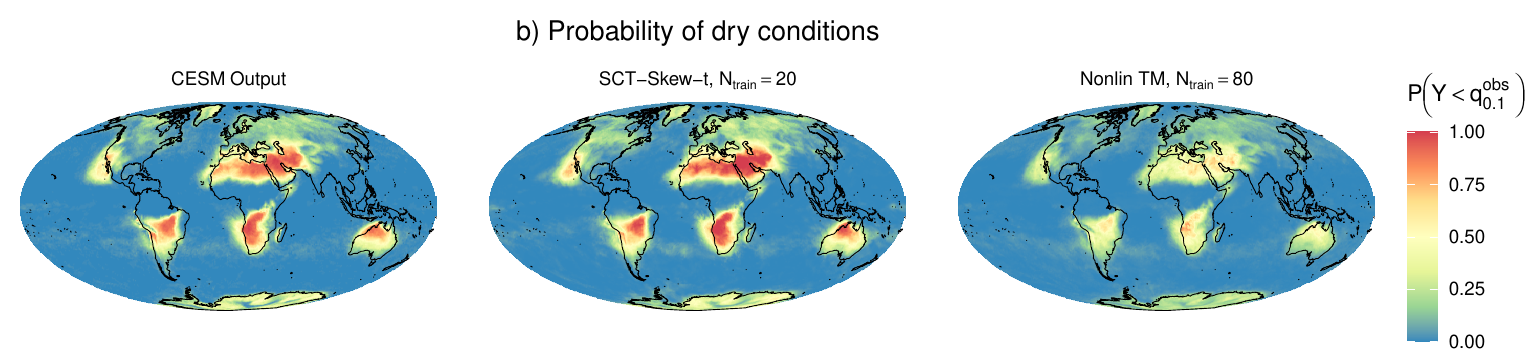}
    \caption{Comparison of location-wise probability of wet (row a) and dry (row b) conditions, defined as values of log precipitation rate above the global $0.95$ quantile and below the global $0.1$ quantile of all 98 fields in the CESM data, respectively ($q_{0.95}^{\text{obs}} \approx -15.81$ and $q_{0.1}^{\text{obs}} \approx -22.78$). We compare the probabilities in the full CESM data (left) to \texttt{SCT-Skew-t} trained on 20 fields (middle) and \texttt{nonlin TM} trained on 80 fields (right).
    \label{fig:PYsmall}}
\end{figure}

Section~F of the Appendix \citep{Brachem2026-SCT-Appendix} provides marginal-only comparisons showing that the fitted marginal model alone tends to oversmooth sharp spatial features, but that adding the transport map restores much of the local structure in generated fields, indicating that the transport map can partially compensate for marginal oversmoothing. Section~H of the Appendix then provides spatial marginal quantile diagnostics that compare model summaries to empirical summaries from the available CESM samples. Because the CESM reference quantiles are estimated from only 98 fields, and because for $\Ntrain=80$ most fields contributing to this reference are also used for training, these maps are interpreted as finite-ensemble calibration diagnostics rather than independent estimates of true marginal quantile error. With this caveat, the diagnostics show that \texttt{SCT-Skew-t} with $\Ntrain=20$ attains average marginal discrepancies superior to \texttt{nonlin TM} with $\Ntrain=80$, while \texttt{SCT-Skew-t} with $\Ntrain=80$ gives the clearest improvement.

\subsubsection{Dependence plausibility}
The preceding visual checks mostly assess marginal plausibility. To also inspect spatial dependence, Figure~\ref{fig:regional-focal-correlations} maps empirical correlations between log precipitation rate at selected focal grid cells and nearby locations. Since the transport map represents nonlinear conditional dependence, these correlations are only a coarse diagnostic rather than a full description of the fitted dependence model. The CESM panels use the available CESM fields, while the model panels are computed from $1000$ generated fields from the corresponding fitted generator.

\begin{figure}
\centering
    \includegraphics[width=\linewidth]{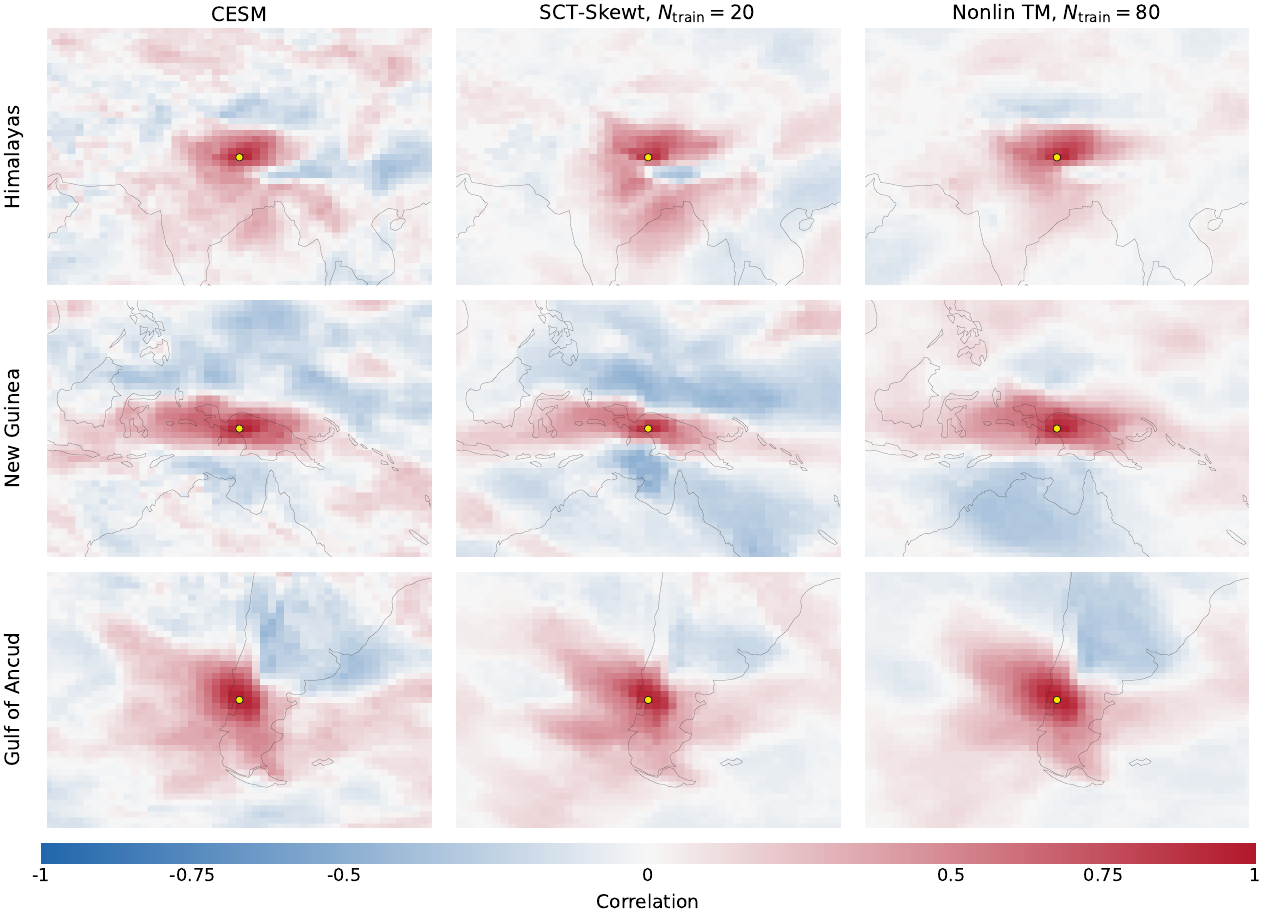}
    \caption{Regional focal-correlation diagnostics for log precipitation rate. Rows show focal locations in the Himalayas, New Guinea, and the Gulf of Ancud; columns show empirical correlations from CESM, \texttt{SCT-Skew-t} trained on $\Ntrain=20$ fields, and \texttt{nonlin TM} trained on $\Ntrain=80$ fields. Each panel covers a $60^\circ$-longitude by $40^\circ$-latitude window centered on the focal grid cell. The yellow point marks the focal location in each row; the corresponding one-based maximin-ordering positions are $16{,}305$, $5$, and $34{,}810$, respectively. Model correlations are computed from $1000$ generated fields. }
    \label{fig:regional-focal-correlations}
\end{figure}

The resulting patterns are physically sensible in the limited but relevant sense that the model-generated fields reproduce regional structures present in the CESM data. In the Himalayas, positive dependence remains strong across the Tibetan Plateau and changes sign at the south-eastern edge. In New Guinea, the positive dependence is organized in an east-west band across the island and nearby Maritime Continent region, with sign changes to the north and south of the band.
For the Gulf of Ancud focal location, the strongest positive correlations are concentrated near Chilo\'e, the Gulf of Ancud, and the adjacent Andean slopes. The conspicuous, nearly right-angled sign reversal is consistent with the local geography: one branch roughly follows the north–south Andes and Chilean coastal archipelago, while the other is roughly aligned with a transition across the Andes from the Pacific-facing slopes to the drier inland side.
Note also that the Gulf of Ancud focal point is Location $34{,}810$ in the maximin-ordering, so the direct conditioning set for its transport map regression component essentially consists of the $m=14$ closest pixels, where $m$ is the number of diagonal entries of $\bsQ_{34{,}810}$ that are non-negligible, operationalized here as being $\geq 0.01$. Notably, the range of the visible marginal correlation extends over a region at least an order of magnitude larger than that covered by the 14 nearest pixels, illustrating that the transport map’s structure effectively propagates long-range dependence through the maximin ordering.

\texttt{SCT-Skew-t} with $\Ntrain=20$ captures these localized and geographically organized patterns about as well as, and in some regions more sharply than, \texttt{nonlin TM} with $\Ntrain=80$. These plots should not be read as evidence that the generator has learned the physical mechanisms themselves: the models can only recover statistical structure that is represented in the training fields. They do show, however, that the generated samples preserve dependence patterns that are geographically interpretable and visible in the CESM ensemble.

Section~G of the Appendix \citep{Brachem2026-SCT-Appendix} provides a complementary global focal-correlation diagnostic showing that the SCT model can also preserve weak but organized long-distance dependence; in this case a positive association near the Maritime Continent together with weaker opposite-signed structure toward the central and eastern equatorial Pacific that is qualitatively consistent with the west-east tropical-Pacific contrasts associated with El Ni\~no-Southern Oscillation (ENSO).

\subsubsection{Inducing locations and runtime} \label{sec:inducing-locs-and-runtime}
Panels c) and d) of Figure~\ref{fig:logscores} compare predictive log scores from \texttt{SCT-Skew-t} when using different numbers of inducing locations $M$ for the Americas data (c) and the whole global grid (d). In both cases, there is no indication of performance improvements to be gained from using more inducing locations; in fact, performance appears to worsen slightly when $M>64$ for the Americas and $M>256$ for the global grid. This sensitivity analysis also informs the interpretation of the spatial marginal diagnostics in the Appendix: increasing the rank of the inducing-point approximation is not supported by the log-score results as a general remedy for remaining spatially structured marginal discrepancies. This is likely a result of the division of labor between the marginal and the dependence model parts: the marginal models, where the inducing locations are relevant, can benefit from the regularizing effect of a high degree of spatial smoothness, as implied by low numbers of inducing points, because the dependence model is well-equipped to capture fine-scale local dependencies in the data. As larger $M$ also substantially increases the computational cost, we use $M=256$ when analyzing the global data. Then, one optimization step takes about 0.7 (0.36) seconds for $\Ntrain=80$ ($\Ntrain=10$), which is insignificant compared to the time cost of about 6 (0.6) seconds per iteration for the dependence model $\calT$ at $\Ntrain=80$ ($\Ntrain=10$). The timing results were obtained on a 14 inch 2021 M1 MacBook Pro using the \texttt{SCT-Gauss} configuration including $\calH$ for the marginal model as described above.
Remarkably, our full model can be trained in minutes on this standard laptop, giving it a small computational footprint relative to the thousands of core-hours required for the original climate simulations.

\section{Conclusions}
\label{sec:conclusion}

We have proposed a scalable generative framework for modeling the internal variability of high-dimensional, non-Gaussian climate fields from limited climate-model ensembles. By composing flexible marginal layers with a sparse Bayesian transport map, our method decouples the estimation of heavy-tailed marginal features from complex spatial dependencies. Crucially, our results demonstrate that this architectural choice enables accurate distributional modeling from small training ensembles, achieving state-of-the-art accuracy with 87.5\% less training data than existing methods. For climate science, this efficiency effectively octuples the computational budget, enabling richer ensemble analyses of tail risks without additional supercomputing costs. While applied to global log-precipitation-rate fields here, the approach is broadly applicable to any massive multivariate setting where data are expensive but structured in a way that allows a distance-based ordering.

A promising direction for extending the model is the inclusion of temporal dynamics and physical drivers (e.g., global mean temperature, emissions levels) or other covariates, such as geographic features (e.g., elevation).
These are natural extensions rather than conceptual obstacles: covariates can enter the parametric marginal layer through distributional regression \citep{Rigby2005-GeneralizedAdditiveModels}, the Gaussian-process priors on marginal parameters can be extended from spatial to spatio-temporal covariance functions, and spatio-temporal or covariate-dependent transport maps \citep{Lei-Cramer2026-ScalableGenerativeModeling} can be combined with the present marginal construction.

 Such extensions would also move the model beyond the fixed-forcing distribution studied here and would make extrapolation uncertainty in the marginal response to covariates more important.

Another important avenue concerns uncertainty quantification.
As discussed in Section~\ref{sec:uncertainty}, our samples are conditional on point estimates for the marginal transformations and transport-map hyperparameters. Future work should develop scalable uncertainty-propagation strategies, especially for threshold-based summaries where uncertainty in the marginal transformations may affect probability statements.

Overall, our proposed framework of scalable composite transformations provides a flexible basis for modeling complex multivariate distributions, and we anticipate that it will support a wide range of future applications.

\footnotesize
\appendix
\numberwithin{figure}{section}
\section*{Acknowledgments}
MK and PFVW were partially supported by NASA's Advanced Information Systems Technology Program (AIST-21). MK's research was also partially supported by National Science Foundation (NSF) Grant DMS--1953005, by industry members of the Center for Interdisciplinary Research on Convective Storms (CIRCS) and NSF award numbers 2517152 and 2517615, and by the Office of the Vice Chancellor for Research at the University of Wisconsin--Madison with funding from the Wisconsin Alumni Research Foundation. JB was partially supported by a scholarship from the German Academic Exchange Service (DAAD) and by the German Research Foundation (DFG) through grant 443179956.

We would like to thank Felix Jimenez and Mark Risser for helpful comments and discussions. During manuscript revision, the authors used OpenAI Codex for coding assistance and language editing of individual paragraphs. All generated outputs were critically reviewed and revised by the authors. The authors maintain sole responsibility for the software, analyses, interpretations, and final manuscript.

\clearpage
\section{Proofs}
\label{app-sec:proofs}

\begin{lemma}[$\calH_i$ passes through $k_1$]
    \label{lemma:pass-k1}
    The function $\calH_i(\cdot | \bsbeta_i)$ fulfills $\calH_i(k_1|\bsbeta_i)=k_1$ for all $\bsbeta_i \in \bbR^{D}$.
\end{lemma}

\begin{proof}
    Due to the localized definition of B-splines, the only nonzero basis functions in $\calH_i(k_1|\bsbeta_i)$ are $B_1$, $B_2$, and $B_3$. Using this fact, and inserting our definitions $k_0 = \gamma_{1,i}$ and $k=\exp(\gamma_{j,i})$ for $j \in \{2,3\}$, we can write
    $\calH_i(k_1|\bsbeta_i) = \sum_{j=1}^3 B_j(k_1)\left(k_0 + \sum_{\ell = 2}^j k\right)$.
    Now observe that, for any $x \in [k_1, k_m]$, B-splines have the property of unity decomposition, meaning that $\sum_{j=1}^JB_j(x) = 1$. Since only the first three bases are nonzero when evaluating at $k_1$, this implies that $\sum_{j=1}^3 B_j(k_1) = 1$, such that we can simplify to $\calH_i(k_1|\bsbeta_i)=k_0+k\sum_{j=1}^4 j B_j(k_1)-k$.
    By way of the recursive definition of B-spline bases \citep[p. 429]{Fahrmeir2013-RegressionModelsMethods}, we can now see that we have $\sum_{j=1}^4 j B_j(k_1) = 2$, such that $\calH_i(k_1|\bsbeta_i) = k_0 + k = k_1$ as desired.
\end{proof}

\begin{lemma}[Unit slope on {$[k_1, a]$} and {$[b, k_m]$}]
    \label{lemma:unitslope-k1a}
    For any $x \in [k_1, a]$, the function $\calH_i(\cdot | \bsbeta_i)$ fulfills $\frac{\partial }{\partial x}\calH_i(x|\bsbeta_i)=1$ for all $\bsbeta_i \in \bbR^{D}$.
\end{lemma}

\begin{proof}
    By Lemma B.1 in \citet{Brachem2025-BayesianPenalizedTransformation}, the derivative of $\calH_i(x|\bsbeta_i)$ for $x \in [k_1, a]$ can be written as
    $\frac{\partial}{\partial x} \calH_i(x | \bsbeta_i) = \sum_{j = 2}^{4} B_{j}^{(2)}(x) \exp(\gamma_{j,i}) / k$, where $B_j^{(2)}$ denotes bases of order two, and we assume the same knot sequence as above. Inserting $\exp(\gamma_{j,i}) = k$ for $j\in \{2, 3, 4\}$, we note that $\frac{\partial}{\partial x} \calH_i(x | \bsbeta_i) = \sum_{j = 2}^{4} B_{j}^{(2)}(x)=1$ for any $x \in [k_1, a]$, by the unity decomposition property of B-spline bases. The proof for $x \in [b, k_m]$ is completely analogous.
\end{proof}

\begin{lemma}[Average slope of one over {$[a, b]$}]
    \label{prop:avgslope-ab}
    The function $\calH_i(x | \bsbeta_i)$ has an average slope of one over the interval $x \in [a, b]$, i.e. for all $\bsbeta_i \in \bbR^D$ it fulfills $\frac{1}{b-a}\int_a^b \frac{\partial}{\partial x} \calH_i(x | \bsbeta_i) \mathrm d x = 1.$
\end{lemma}

\begin{proof}
    To reduce notational burden, we suppress the location index $i$ for the coefficients $\beta$ and $\gamma$ in this proof, which can be done without loss of generality.
    By Lemma B.2 in the appendix of \citet{Brachem2025-BayesianPenalizedTransformation}, the average derivative in a segment $x \in [k_j, k_{j+1}]$ of a monotonically increasing B-spline defined on equidistant knots as given in Eq. (3) of the main text is known in closed form as
    \begin{equation}
        \frac{1}{k} \int_{k_j}^{k_{j+1}} \frac{\partial}{\partial x} \calH_i(x) \mathrm d x = \frac{1}{k} \left(
        \frac{1}{6} \exp(\gamma_{j+1}) + \frac{2}{3} \exp(\gamma_{j+2}) + \frac{1}{6} \exp(\gamma_{j+3})
        \right).
    \end{equation}
    Since $a = k_2$ and $b=k_{m-1}$, we can obtain the average slope over the interval $[a, b]$ by averaging over the local averages of all function segments $[k_j, k_{j+1}]$ for $j=2, 3, \dots, m-2$:
    \begin{align}
        \frac{1}{b-a} \int_{a}^{b} \frac{\partial}{\partial x} \calH_i(x) \mathrm d x
        & = \frac{1}{k_{m-1}-k_2} \int_{k_2}^{k_{m-1}} \frac{\partial}{\partial x} \calH_i(x) \mathrm d x             \\
        & = \frac{1}{k(m-3)} \sum_{j=2}^{m-2} \int_{k_j}^{k_{j+1}} \frac{\partial}{\partial x} \calH_i(x) \mathrm d x \\
         & =\frac{1}{k(m-3)} \sum_{j=2}^{m-2} \left(
        \frac{1}{6} \exp(\gamma_{j+1}) + \frac{2}{3} \exp(\gamma_{j+2}) + \frac{1}{6} \exp(\gamma_{j+3})
        \right)                                                                   \\
         & = \frac{1}{k(m-3)} \left(2k + \sum_{j=5}^{J-3} \exp(\gamma_j) \right).
    \end{align}
    Note that $k_{m-1}-k_2 = k(m-3)$.
    For rearranging the sum on the last line, we observe that we fixed $\gamma_j = \log(k)$ for $j \in \{2,3,4\} \cup \{J-2, J-1, J\}$. The simplified sum on the last line then becomes possible by making use of the fact that the sum contains the terms $\exp(\gamma_3)/6 + 2\exp(\gamma_4)/3 + \exp(\gamma_4)/6 = k$ and $\exp(\gamma_{J-1})/6 + 2\exp(\gamma_{J-2})/3 + \exp(\gamma_{J-2})/6 = k$. For the remaining $\gamma_5, \gamma_6, \dots, \gamma_{J-3}$, the fractions sum up to one.
    Now recall that, for $j \in \{5,6, \dots, J-3\}$, we have defined
    $\gamma_{j} = \beta_{j-4} -
        \log \left(
        ((m-5)k)^{-1}\sum_{d=1}^D \exp(\beta_{d})
        \right)$,
    with $\beta_d \in \bbR$ for $d = 1, \dots D$ and $D = J - 7$.
    We can thus write
    \begin{align}
        \sum_{j=5}^{J-3} \exp(\gamma_j) & =
        \sum_{d=1}^D \exp \left(
        \beta_d -
        \log \left(
            \frac{\sum_{d=1}^D \exp(\beta_d)}{(m-5)k} \right)
        \right)                                               \\
                                        & =
        \sum_{d=1}^D \exp (\beta_d)
        \left(
        \frac{(m-5)k}{\sum_{d=1}^D \exp(\beta_d)} \right) \\
        &= (m-5)k
    \end{align}
    and insert this expression back to arrive at $(b-a)^{-1} \int_{a}^{b} \frac{\partial}{\partial x} \calH_i(x) \mathrm d x  = (k(m-3))^{-1}k(m-5 + 2) = 1$,
    which concludes the proof.
\end{proof}

\begin{lemma}[$\calH_i$ passes through $a$ and $b$]
    \label{lemma:pass-ab}
    The function $\calH_i(\cdot | \bsbeta_i)$ fulfills $\calH_i(a|\bsbeta_i)=a$ and $\calH_i(b|\bsbeta_i)=b$ for all $\bsbeta_i \in \bbR^{D}$.
\end{lemma}

\begin{proof}
    By Lemma \ref{lemma:pass-k1}, we have $\calH_i(k_1 | \bsbeta_i) = k_1$, and by Lemma \ref{lemma:unitslope-k1a}, the slope of $\calH_i(x|\bsbeta_i)$ for all $x \in [k_1, a]$ is one. As a result, we directly get $\calH_i(a|\bsbeta_i)=a$.
    Further, by Lemma~\ref{prop:avgslope-ab}, the average slope of $\calH_i(x|\bsbeta_i)$ over the interval $[a,b]$ is one, which, given that $\calH_i(a|\bsbeta_i)=a$ as we just noted, implies that $\calH_i(b|\bsbeta_i)=b$.
\end{proof}

\begin{proposition}[$\calH_i$ is identity outside a compact set]
    \label{prop:identity-tails}
    The function $\calH_i$ fulfills $\calH_i(\tilde y_i|\bsbeta_i)=\tilde y_i$ for all $\tilde y_i \in (-\infty, a] \cup [b, \infty)$ and for all $\bsbeta_i \in \bbR^{D}$.
\end{proposition}

\begin{proof}[Proof of Proposition~\ref{prop:identity-tails}]
    By Lemma \ref{lemma:pass-ab}, we have $\calH_i(a|\bsbeta_i)=a$ and $\calH_i(b|\bsbeta_i)=b$, and by Lemma \ref{lemma:unitslope-k1a}, $\calH_i(x|\bsbeta_i)$ has a unit slope for all $x \in [k_1,a] \cup [b, k_m]$. This means that $\calH_i(x|\bsbeta_i) = x$ for all $x \in [k_1,a] \cup [b, k_m]$. The statements above all hold for all $\bsbeta_i \in \bbR^{D}$. Since $\calH_i(x|\bsbeta_i) = x$ for $x < k_1$ and $x > k_m$ by definition, this concludes the proof.
\end{proof}

From a technical perspective, this result is important, because it shows that the transformation function's transition back to identity at both tails happens in the B-spline function segment, which means in turn that the transition is continuous and twice continuously differentiable, and that the function as a whole is continuous and twice continuously differentiable.
From a modeling perspective, it is important, because it allows us to think of $\calH_i$ as a flexible increasing B-spline on $(a, b)$ that transitions smoothly to the identity function beyond the flexible subdomain.
Since $\calH_i$ is a strictly monotonically increasing B-spline on $[k_1, k_m]$ by construction, this proposition also implies that $\calH_i(\tilde y_i | \bsbeta_i)$ is strictly monotonically increasing in $\tilde y_i$ for all $\tilde y_i \in \bbR$ and $\bsbeta_i \in \bbR^{D}$. Further, it allows us to note that the derivative of $\calH_i(\tilde y_i | \bsbeta_i)$ with respect to $\tilde y_i$ is a second-order B-spline
$(\partial / \partial \tilde y_i) \calH_i(\tilde y_i | \bsbeta_i) =
    \sum_{j=2}^{J} B_j^{(2)}(\tilde y_i)\exp(\gamma_{j,i}) / k$ for $\tilde y_i \in [k_1, k_m]$ and $1$ else, where $B^{(2)}_j$ are second-order B-spline bases defined on the same knots as Eq. (3) of the main text. Note that $(\partial / \partial \tilde y_i) \calH_i(\tilde y_i | \bsbeta_i)$ depends on $\bsbeta_i$ through $\gamma_{5,i}, \dots, \gamma_{J-3,i}$ as given in Eq. (4) of the main text.

\begin{proof}[Proof of Theorem~2.1]
    We can use the identities $\beta_{i,1} = \beta_{i,2} = \dots = \beta_{i,D} = \beta_i$ to simplify the definition of $\gamma_{j,i}$ for $j \in \{5, 6, \dots, J-3\}$ to $\gamma_{j,i} = \log(m-5) - \log(D) + \log(k)$.
    Since $m = D+5$, this expression reduces to $\gamma_{j,i} = \log(k)$, such that we now have $\gamma_{j,i} = \log(k)$ for \textit{all} $j = 2, \dots, J$. This allows us to apply Lemma \ref{lemma:unitslope-k1a} separately to all function segments $[k_j, k_{j+1}]$ for $j = 1, \dots, m-1$, showing that $\calH_i(x | \bsbeta_i)$ has a constant unit slope for all $x \in [k_1, k_m]$. When combining this observation with Proposition~\ref{prop:identity-tails}, which shows that $\calH_i(x | \bsbeta_i)$ is the identity function for $x \in (-\infty, k_2] \cup [k_{m-1}, \infty)$, the proof is complete.
\end{proof}

\begin{proof}[Proof of Theorem~2.2]
    Since both $\calH_i$ and $\calG_i$ are strictly monotonically increasing, the assumption $(\calH_i \circ \calG_i)(y_i) \sim \calN(0, 1)$ implies that
    $\bbP(y_i \leq t)  = \bbP\bigl((\calH_i \circ \calG_i)(y_i) \leq (\calH_i \circ \calG_i)(t)\bigr) = (\Phi \circ \calH_i \circ \calG_i)(t)$.
    Further, since $t \in (-\infty, q_a] \cup [q_b, \infty)$, we have $\calG_i(t) \in (-\infty, a] \cup [b, \infty)$.
    Applying Proposition \ref{prop:identity-tails}, this means that $(\calH_i \circ \calG_i)(t) = \calG_i(t)$, which now allows us to write $\bbP(y_i \leq t) = (\Phi \circ \calG_i)(t) = (\Phi \circ \Phi^{-1} \circ F_i)(t) = F_i(t)$,
    by substituting the definition of $\calG_i$.
\end{proof}

\clearpage
\section{Improper prior induced by a flat inverse-softplus parameterization}
\label{app:prior-tau2}

Let $\eta \in \mathbb{R}$ and define
$$
    \tau^2 = \{\operatorname{softplus}(\eta)\}^2 = \{\log(1 + \exp(\eta))\}^2, \qquad \tau > 0 .
$$
Assume an (improper) constant prior on the latent variable, $p(\eta) \propto 1$.
The inverse transformation is
$$
    \eta = \log\left(\exp(\tau) - 1\right),
$$
with derivative
$$
    \frac{d \eta}{d\tau} = \frac{\exp(\tau)}{\exp(\tau)-1}
    = \frac{1}{1 - \exp(-\tau)} .
$$
By change of variables, the induced improper prior on $\tau$ is therefore
$$
    p(\tau)
    \propto \left| \frac{d \eta}{d\tau} \right|
    = \frac{1}{1 - \exp(-\tau)},
    \qquad \tau > 0,
$$
and the induced improper prior on $\tau^2$ is
$$
    p(\tau^2)
    \propto p(\tau) \left| \frac{d \tau}{d \tau^2} \right|
    = \frac{1}{1 - \exp(-\sqrt{\tau^2})} \frac{1}{2 \sqrt{\tau^2}},
    \qquad \tau^2 > 0 .
$$
As $\tau^2 \to 0$, we also have $\tau \to 0$, and consequently $\exp(-\tau) \to 1$. Thus, in the improper prior above, both fractions become large as $\tau^2 \to 0$, such that small values of $\tau^2$ are favored by this prior. At the same time, as $\tau^2 \to \infty$, we have $\exp(-\tau) \to 0$, so that $(1 - \exp(-\sqrt{\tau^2}))^{-1} \to 1$. Then, for any $a > 0$, the tail probability mass implied by this prior as $\tau^2 \to \infty$ is
$$
    \int_a^\infty p(\tau^2)\, \mathrm d \tau^2 \geq \int_a^\infty \frac{1}{2\sqrt{\tau^2}} \, \mathrm d \tau^2 = \infty,
$$
which diverges. In summary, the flat prior on $\eta = \log(\exp(\tau) - 1)$ induces an improper prior on $\tau$ and $\tau^2$ that strongly favors small variances, but only weakly penalizes large variances due to a heavy right tail.

\clearpage
\section{Dominant computational costs}
\label{app:complexity}

\autoref{tab:complexity} summarizes the dominant computational complexity of the main model components assuming $\Ndim$ locations, $\Nobs$ independent samples of the full field, $M$ inducing locations for the low-rank approximations in the marginal models, $P$ parameters in each parametric marginal model $\calG_i$, and $D$ flexible spline-parameters in each semi-parametric marginal model $\calH_i$. Lower-order terms and constant factors are omitted. Spline evaluations in $\calH$ are accelerated using linear interpolation on a precomputed regular grid of size $G=1000$, reducing basis-function evaluation to a binary search followed by interpolation. The derivative of the spline reuses the same grid and search results, so derivative evaluation introduces only an additional matrix multiplication and therefore has the same asymptotic dependence on $\Nobs$, $\Ndim$, and $D$. See \citet{Katzfuss2023-ScalableBayesianTransport} for details on the transport map complexity. The semi-parametric marginal layer is inverted numerically using cubic interpolation on a regular grid via the Python library Interpax \citep{Conlin2026-Interpax}.

Assuming $\Ndim > M > \Nobs > D > m > \log G > P$, the overall time complexity of a single evaluation of the unnormalized log posterior is $\mathcal{O}(\Ndim D(M + \Nobs) + M^3P)$ for our marginal model and $\mathcal{O}(\Ndim\Nobs^3)$ for the Bayesian transport-map dependence model.

\begin{table}[h]
\centering
\caption{Dominant computational complexity of the main model components; $m$ denotes the size of the conditioning set in the posterior transport map, see Section~3.2 of the main text. The overview assumes a Gaussian model for $\calG$.}
\label{tab:complexity}
\footnotesize
\begin{tabular}{ll}
\toprule
Component & Complexity \\ 
\midrule

\multicolumn{2}{l}{\textsf{\textbf{Marginal Model}}} \\[4pt]

Prior for $\bszeta$ (low-rank, whitened)
& $\mathcal{O}(PM)$ \\[4pt]

Approximation of $\bszeta$ 
& $\mathcal{O}(P\Ndim M + PM^3)$ \\[4pt]

Prior for $\bsbeta$ (low-rank, whitened)
& $\mathcal{O}(DM)$ \\[4pt]

Approximation of $\bsbeta$
& $\mathcal{O}(\Ndim MD + M^3)$ \\[4pt]

Evaluation of $\calG$
& $\mathcal{O}(\Nobs\Ndim)$ \\[4pt]

Evaluation of $\partial \calG / \partial y$
& $\mathcal{O}(\Nobs\Ndim)$ \\[4pt]

Evaluation of $\calH$
& $\mathcal{O}(\Nobs\Ndim\log(G) + \Nobs\Ndim D)$ \\[4pt]

Evaluation of $\partial \calH / \partial \tilde y$
& $\mathcal{O}(\Nobs\Ndim D)$ \\[4pt] \addlinespace

\multicolumn{2}{l}{\textsf{\textbf{Dependence Model}}} \\[4pt]

Evaluation of Gaussian reference density
& $\mathcal{O}(\Nobs\Ndim)$ \\[4pt]

Evaluation of  $p(\hat{\tilde \bfZ}|\bstheta^\calT)$
& $\mathcal{O}(\Ndim\Nobs^3)$ \\[4pt]

Evaluation of $\widetilde \calT$
& $\mathcal{O}(\Ndim(\Nobs^3 + m\Nobs^2))$ \\[4pt]

Application of $\widetilde \calT$
& $\mathcal{O}(\Ndim (\Nobs^2 + m\Nobs))$ \\[4pt] \addlinespace

\multicolumn{2}{l}{\textsf{\textbf{Inversion}}} \\[4pt]

Inverse $\calG^{-1}$
& $\mathcal{O}(\Nobs\Ndim)$ \\[4pt]

Inverse $\calH^{-1}$
& $\mathcal{O}(\Nobs\Ndim)$ \\[4pt]

Inverse transport map $\widetilde{\calT}^{-1}$
& $\mathcal{O}(\Ndim (\Nobs^2 + m\Nobs))$ \\

\bottomrule
\end{tabular}
\end{table}

\clearpage
\section{Potential use for compressed storage}
\label{app:conditional-generation-compression}

From a storage perspective, the marginal layers are comparatively compact after fitting, because they are represented by parameter fields of a total size of $M(D+P)$ elements. The fitted Bayesian transport map is less compact: its closed-form posterior map is evaluated through GP-prediction equations that depend on the training ensemble and on an $\Nobs \times \Nobs$ matrix for each ordered location \citep[Proposition~1]{Katzfuss2023-ScalableBayesianTransport}. Thus, an implementation must either retain the training ensemble and recompute or refactor these matrices as needed, or store their Cholesky factors or equivalent linear-solve quantities.
Compared with storage-oriented stochastic generators that retain a compact parameter set for unconditional generation of new ensemble-like simulations \citep{Huang2023-SavingStorageClimate,Song2024-EfficientStochasticGenerators}, this dependence-layer storage requirement makes our fitted model less suitable as a standalone compressed representation of the training ensemble.

Nevertheless, once our surrogate model has been trained, additional spatial fields can be represented in compressed stochastic form by storing only a small number of leading ordered locations or map coefficients and then drawing conditional samples for the remaining locations, as \citet{Katzfuss2023-ScalableBayesianTransport} illustrate for Bayesian transport maps. This use is closer in spirit to conditional compression, where selected information from a target field is stored and the remaining field is modeled conditionally \citep{Guinness2018-CompressionConditionalEmulation}. While the fitted model therefore suggests a secondary use case for possible stochastic compression of additional fields beyond the training data, we do not present it here as a standalone compression method for replacing the original training ensemble.

Figure~\ref{app-fig:conditional-samples-compression} illustrates this idea. For these plots,
we take one held-out global log-precipitation-rate field and draw new samples, conditioning on the observed values at the first $50$, $500$, or $5000$ maximin-ordered locations of the test field.
Since early locations in the ordering are spatially dispersed, even conditioning on just the first $50$ locations constrains broad spatial features, while larger conditioning sets preserve increasingly local information. For the trimmed global grid with $54{,}722$ locations, retaining $50$, $500$, or $5000$ locations corresponds to approximately $0.09\%$, $0.91\%$, or $9.14\%$ of the field values before accounting for fitted-model overhead. Depending on the required fidelity to the observations, additional fields can therefore be stored at a small fraction of their size. Note that this is a stochastic conditional completion, not a deterministic reconstruction of the held-out field.

\begin{figure}[htb]
    \centering
    \includegraphics[width=\linewidth]{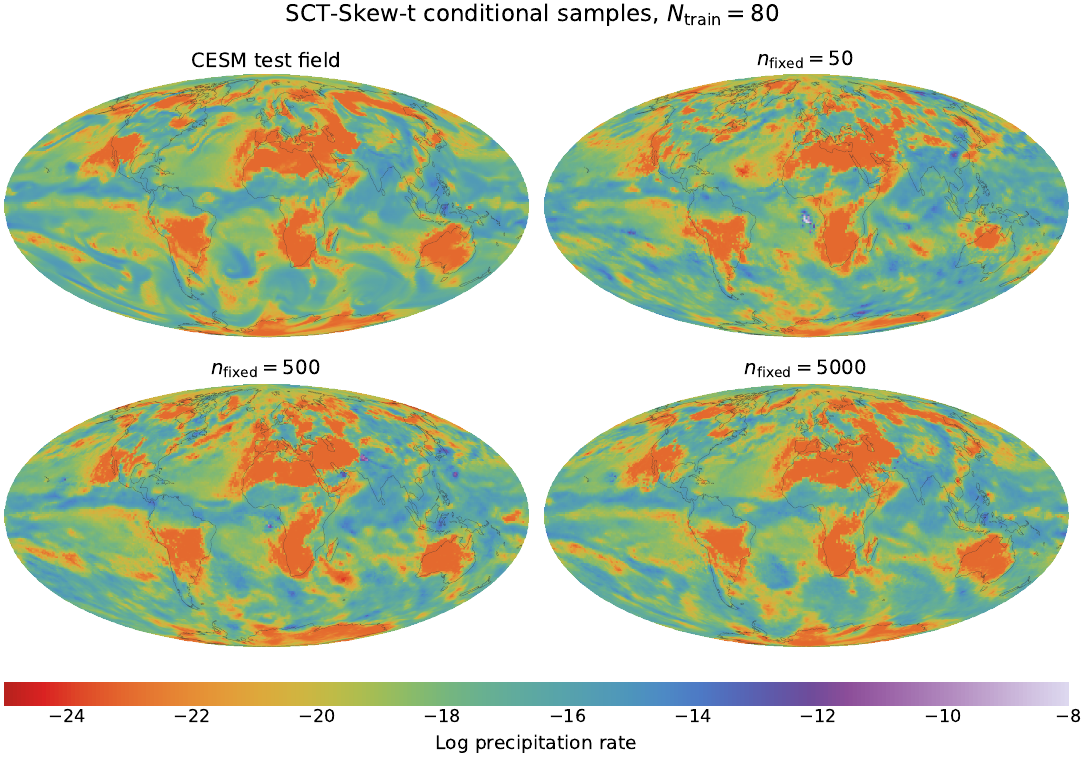}
    \caption{Illustration of conditional sampling as a possible stochastic-compression use case. The bottom right panel shows one held-out CESM log-precipitation-rate field. The remaining panels show conditional samples from \texttt{SCT-Skew-t} trained on $\Ntrain=80$ fields after fixing the first 50, 500, or 5000 locations in the maximin ordering.}
    \label{app-fig:conditional-samples-compression}
\end{figure}

\clearpage
\section{Adjusted log scores}
\label{app:shifted-logscores}

In Figure~5 of the main text, we show log scores taken from Figure 10 in \citet{Katzfuss2023-ScalableBayesianTransport}, but shift them for a fair comparison.
This shift is necessary, because the authors applied their models to a version of the data that was preprocessed by centering and scaling at each location based on the location-specific sample mean and standard deviation of the full sample. This step somewhat resembles our definition of $\calG$ in \texttt{SCT-Gauss}, but without spatial smoothing. Since the preprocessing step was invariant among the models compared by \citet{Katzfuss2023-ScalableBayesianTransport}, it could be ignored in the computation of the log scores. Here, however, we consider $\calG$ as an integral part of the model and therefore have to account for it in the computation of the log scores. Consequently, we have to adjust the log scores for the alternative models by computing $-\text{E}[\log p(\bsy^*|\bfY)] + \sum_{i=1}^\Ndim \log s_i$, where $s_i$ is the sample standard deviation over $y_{i,1}, \dots, y_{i,\Nobs}$ at location $i$ and $-\text{E}[\log p(\bsy^*|\bfY)]$ is the log score reported by \citet{Katzfuss2023-ScalableBayesianTransport}. The log scores reported for our model runs are similarly adjusted for our z-standardization preprocessing step, but note that our z-standardization does not use location-specific means and standard deviations, but the global mean and standard deviation of the training sample.

\clearpage
\section{Marginal oversmoothing and transport-map compensation}
\label{app:marginal-oversmoothing-diagnostics}

\FloatBarrier
To separate marginal-layer smoothing from the effect of the dependence model, Figures~\ref{app-fig:oversmoothing-mollweide} and~\ref{app-fig:oversmoothing-regional-himalayas} compare three sets of marginal summaries: the CESM fields, samples from marginal-only models, and samples from the full SCT model. The marginal-only model, denoted \texttt{HG} in the figures, uses the fitted marginal transformations $\calH \circ \calG$ without the transport map. For both \texttt{HG} and \texttt{SCT-Skew-t}, the displayed summaries are computed from $1000$ generated fields at $\Ntrain=20$ and $\Ntrain=80$; the CESM summaries are computed from the available CESM fields.

These comparisons confirm that the potential for marginal oversmoothing is strong in the marginal-only models. The \texttt{HG} rows capture broad large-scale precipitation regimes, but they smooth sharp local transitions and regional texture, most visibly near complex topography and land-sea contrasts. The full \texttt{SCT-Skew-t} generator, however, restores much of this spatial detail in the displayed quantiles and mean, with the $\Ntrain=80$ fit giving the cleaner version. Thus, the transport map does more than add dependence: it also compensates for some residual marginal oversmoothing after the inverse marginal transformations. The regional Himalayas zoom in Figure~\ref{app-fig:oversmoothing-regional-himalayas} makes this effect particularly clear, because the marginal-only summaries are visibly rounded across the mountain region, whereas the full generator recovers more of the spatially abrupt structure visible in the CESM summaries.

\begin{figure}[hb!]
    \centering
    \includegraphics[width=\linewidth]{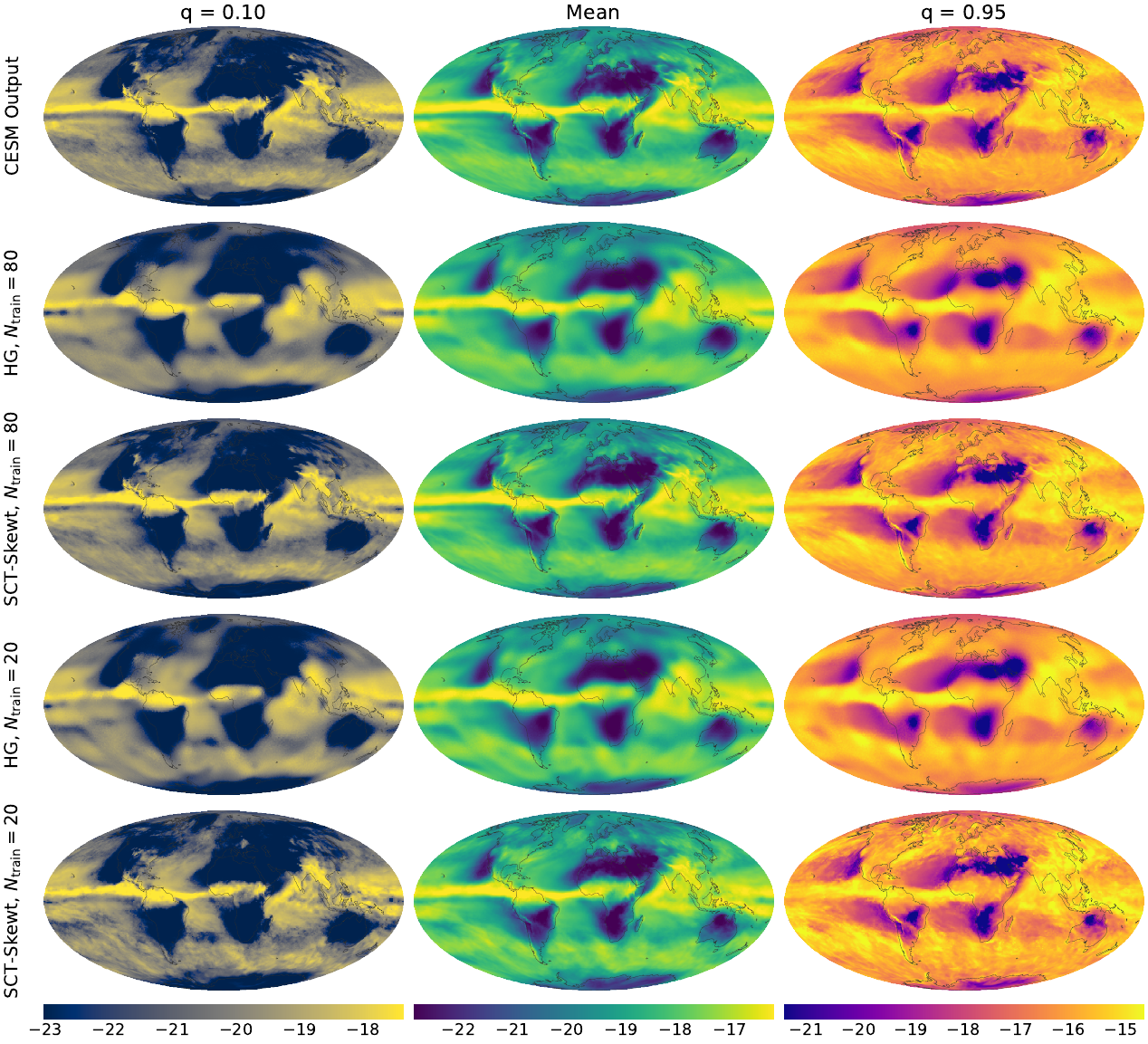}
    \caption{Global comparison of marginal summaries for CESM output, the marginal-only \texttt{HG} model, and the full \texttt{SCT-Skew-t} generator (both of the latter two with $\Ntrain=80$ and $\Ntrain=20$). Columns show the empirical $0.10$ quantile, mean, and empirical $0.95$ quantile of log precipitation rate at each location. The model rows are based on $1000$ generated fields. The marginal-only \texttt{HG} model captures broad spatial regimes but visibly smooths localized structure; the full \texttt{SCT-Skew-t} generator restores much of the CESM-like spatial detail.}
    \label{app-fig:oversmoothing-mollweide}
\end{figure}

\begin{figure}[p]
    \centering
    \includegraphics[width=\linewidth]{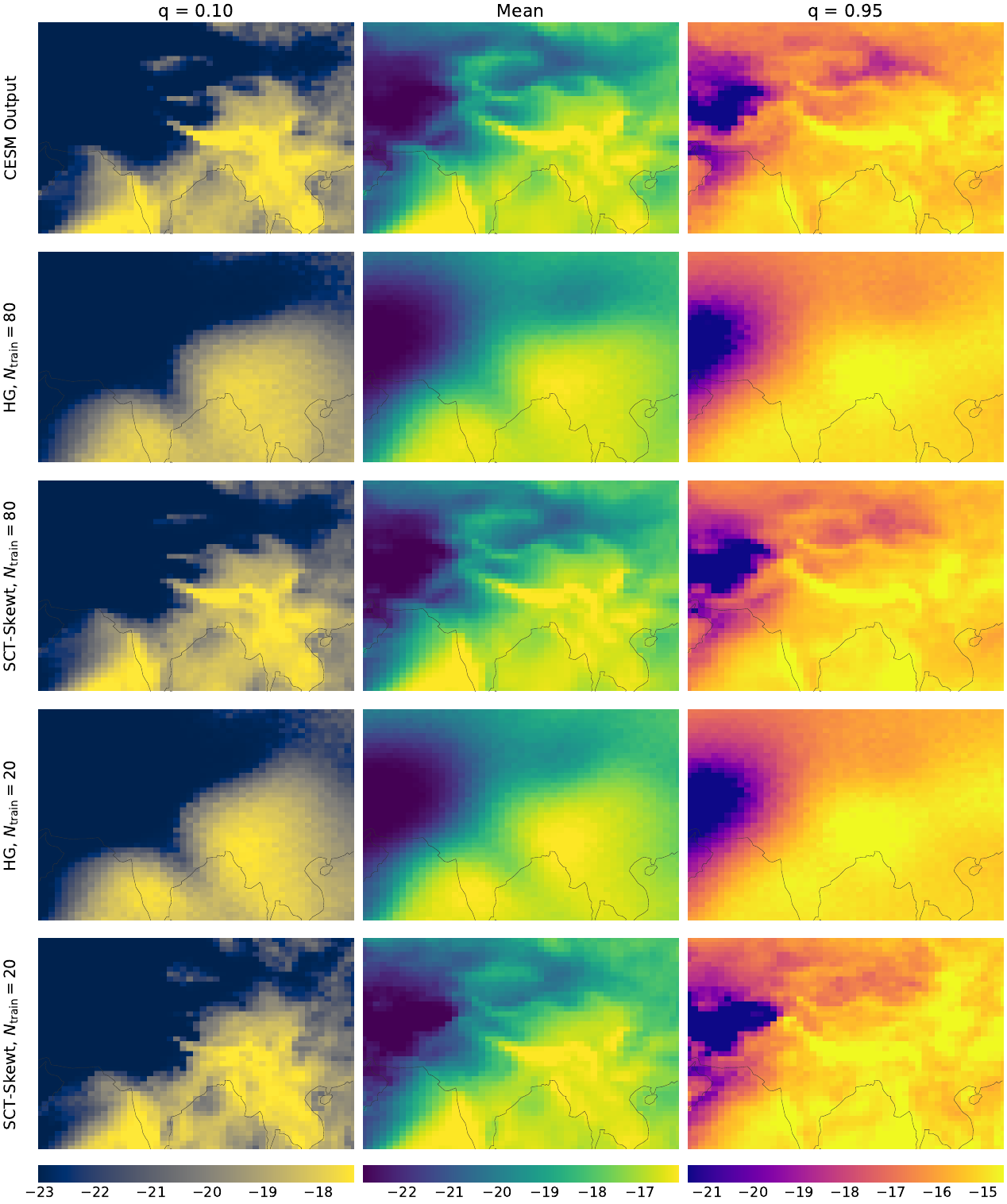}
    \caption{Regional version of Figure~\ref{app-fig:oversmoothing-mollweide} for the Himalayas and surrounding region. The marginal-only \texttt{HG} model smooths the abrupt spatial transitions associated with the mountain region, whereas the full \texttt{SCT-Skew-t} generator recovers more of the local CESM structure in the lower quantile, mean, and upper quantile.}
    \label{app-fig:oversmoothing-regional-himalayas}
\end{figure}

\FloatBarrier
\section{Long-distance focal-correlation diagnostic}
\label{app:long-distance-focal-correlation}

To complement the regional focal-correlation diagnostics in the main manuscript, Figure~\ref{app-fig:world-focal-correlation-pacific} applies the same idea on a global Pacific-centered map. The focal grid cell lies near New Guinea in the Maritime Continent, at longitude $137.5^\circ\mathrm{E}$ and latitude $4.24^\circ\mathrm{S}$, and is the fifth location in the maximin ordering. The pattern is similarly present for nearby locations that appear much later in the ordering (not shown). Each panel shows the empirical correlation between log precipitation rate at this focal location and all other grid cells. The CESM panel is computed from the 98 available July~1 CESM fields, while the model panels are computed from $1{,}000$ generated fields from the corresponding fitted generators.

This diagnostic illustrates that the sparse transport-map construction is not limited to local dependence. In the CESM panel, the focal location has a localized positive association around the Maritime Continent and weak, noisy remote structure across the tropical Pacific. The fitted generators reproduce the local correlation structure. Additionally, the model samples retain a remote contrast toward the central and eastern equatorial Pacific. The remote correlations are modest in magnitude, but they are directionally consistent across CESM and the generated samples.

The broad tropical-Pacific structure in Figure~\ref{app-fig:world-focal-correlation-pacific} provides one way to relate this diagnostic to known physical modes of variability. A positive association near the Maritime Continent together with weaker opposite-signed structure toward the central and eastern equatorial Pacific is qualitatively consistent with the west-east tropical-Pacific contrasts associated with El Ni\~no-Southern Oscillation (ENSO). In ENSO, changes in sea-surface temperature (SST) anomalies and atmospheric circulation link the western Pacific warm-pool region to the central and eastern equatorial Pacific. This comparison is only qualitative, however. The focal variable here is a single-grid-cell July~1 log precipitation rate, not an SST-based ENSO index or a regional ENSO average, and single-day log-precipitation-rate fields are strongly weather-dominated. Moreover, our fitted generators model a fixed distribution of spatial fields rather than a time-evolving process, whereas ENSO variability evolves over months to years. Thus, the figure cannot be read as evidence that the model has learned ENSO dynamics. It does suggest that the fitted generators can preserve a physically plausible long-distance log-precipitation-rate correlation pattern in the tropical Pacific, including a weak ENSO-consistent west-east contrast visible in the CESM fields.

\begin{figure}[htb]
    \centering
    \includegraphics[width=\linewidth]{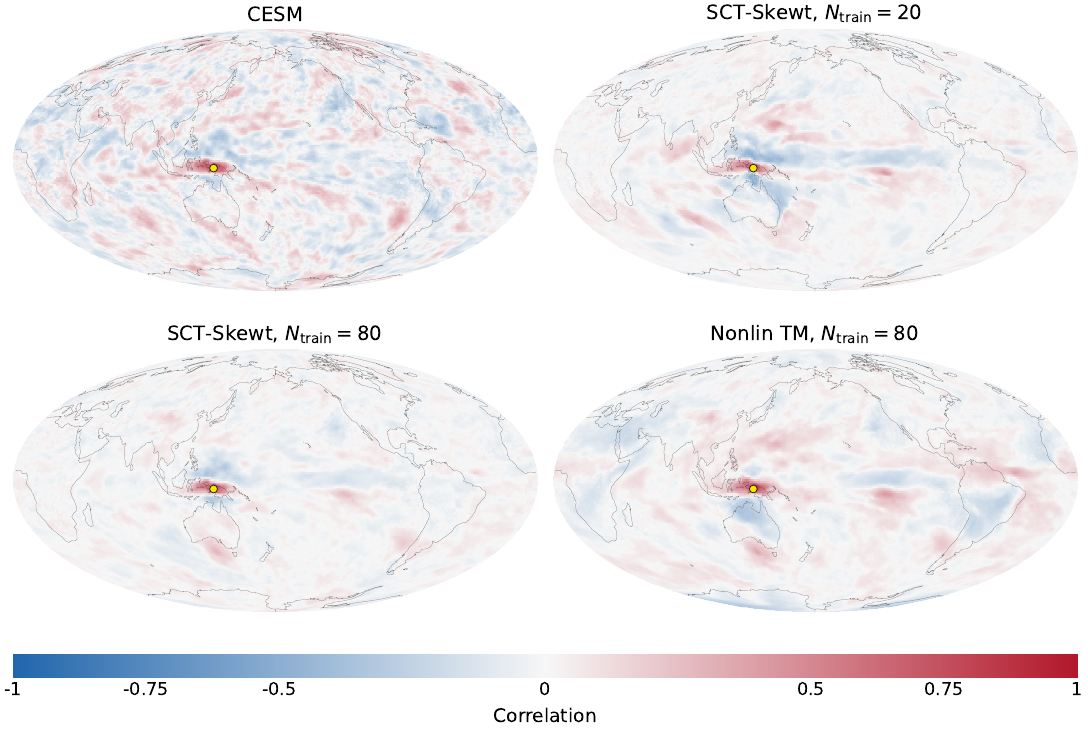}
    \caption{Global focal-correlation diagnostic for a western-Pacific / Maritime Continent focal location near New Guinea. Colors show empirical correlations between log precipitation rate at the focal grid cell, marked in yellow, and log precipitation rate at all other grid cells on a Pacific-centered Mollweide projection. Panels compare CESM, \texttt{SCT-Skew-t} with $\Ntrain=20$, \texttt{SCT-Skew-t} with $\Ntrain=80$, and \texttt{nonlin TM} with $\Ntrain=80$. The model panels are based on $1000$ generated fields.}
    \label{app-fig:world-focal-correlation-pacific}
\end{figure}

\clearpage
\section{Spatial marginal quantile diagnostics}
\label{app:marginal-quantile-diagnostics}

\FloatBarrier
To assess marginal fit systematically across space, Figure~\ref{app-fig:marginal-quantile-diagnostic} maps location-wise marginal quantile discrepancies. For each location $i$, we compute empirical quantiles $\hat q^{\text{CESM}}_{i,p}$ from the CESM fields and $\hat q^{\text{model}}_{i,p}$ from $1000$ fields generated by the fitted surrogate model, using 20 equally spaced probability levels $p \in [0.05, 0.95]$. We then summarize the discrepancy by
\[
    D_i =
    \frac{1}{20 \times s_{\text{global}}}
    \sum_{p \in \mathcal Q}
    \left|
    \hat q^{\text{model}}_{i,p}
    -
    \hat q^{\text{CESM}}_{i,p}
    \right|,
\]
where $\mathcal Q$ denotes this probability grid and $s_{\text{global}}$ is the global sample standard deviation of log precipitation rate. Thus, a value of $0.5$ means that the model's marginal quantiles differ from the CESM quantiles by about half a global standard deviation on average.

This diagnostic should be interpreted as a finite-ensemble calibration check rather than as an independent estimate of the true marginal quantile error. The empirical CESM quantiles are computed from only 98 available fields and therefore have sampling variability, especially near the tails. In addition, for fits with $\Ntrain=80$, most fields contributing to the empirical CESM reference are also used for training. The maps are nevertheless informative for comparing whether different generators reproduce the spatially varying marginal summaries visible in the available CESM ensemble; out-of-sample predictive performance is assessed separately by the log scores in the main text.

The spatial diagnostic is consistent with the interpretation from the log scores and threshold maps. \texttt{SCT-Skew-t} with $\Ntrain=20$ is comparable to \texttt{nonlin TM} trained on $\Ntrain=80$ fields: the mean (median) value of $D_i$ is $0.145$ ($0.121$) for \texttt{SCT-Skew-t} and $0.157$ ($0.143$) for \texttt{nonlin TM}, and \texttt{SCT-Skew-t} has a lower discrepancy at 58\% of locations. At $\Ntrain=80$, the improvement is clearer: \texttt{SCT-Skew-t} has lower discrepancies than \texttt{nonlin TM} at 95\% of locations, and the mean discrepancy is reduced from $0.157$ to $0.073$. At the same time, we note that the remaining \texttt{SCT-Skew-t} discrepancies for $\Ntrain=20$ show some spatial structure, which motivates the signed single-quantile diagnostics below.

The residual spatial structure, particularly for $\Ntrain=20$, should therefore be interpreted as structured marginal misfit rather than direct evidence of oversmoothing in the full generator. In light of the preceding marginal-only diagnostics, local smoothing by the marginal layer remains one possible contributor, but the transport map can also compensate for that smoothing and may introduce locally adaptive corrections. With only a small training ensemble, coherent signed residuals can therefore also reflect local overfitting to the training fields or missing geographical and physical information. We do not interpret these residual patterns as evidence that a larger inducing set would solve the problem, since increasing $M$ did not improve log scores. Instead, reducing the remaining structured discrepancies may require stronger regularization of local corrections in the transport map or the inclusion of additional physical or geographical information via covariates in future extensions.

\begin{figure}[bt]
\centering
    \includegraphics[width=\linewidth, trim={0 0 0 20bp}, clip]{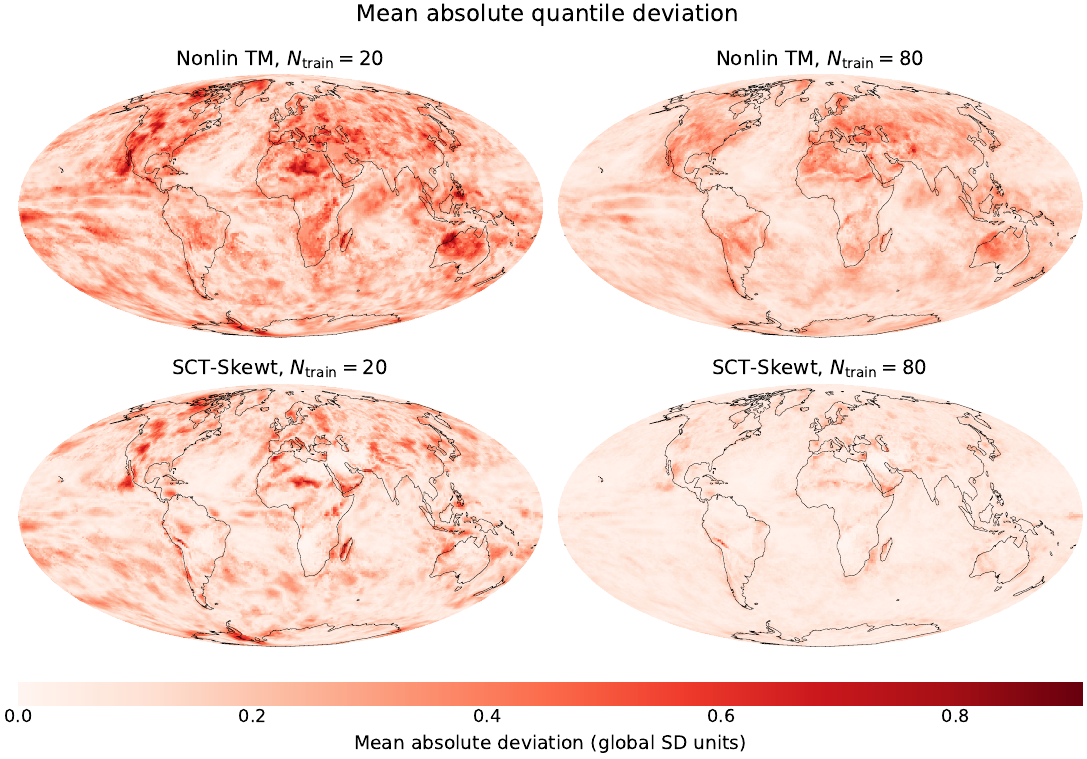}
    \caption{Spatially resolved marginal quantile discrepancies. At each location, the plotted value is the average absolute difference between empirical CESM quantiles and empirical model-sample quantiles over 20 equally spaced probability levels between $0.05$ and $0.95$, scaled by the global sample standard deviation of log precipitation rate. Smaller values indicate closer marginal fit.}
    \label{app-fig:marginal-quantile-diagnostic}
\end{figure}

Figure~\ref{app-fig:marginal-quantile-diagnostic-levels} shows signed location-wise marginal quantile differences at the probability levels $p=0.10$ and $p=0.95$. For a fixed probability level $p$, the diagnostic is
$(\hat q^{\text{model}}_{i,p}-\hat q^{\text{CESM}}_{i,p})/s_{\text{global}}$,
where $s_{\text{global}}$ is the global sample standard deviation of log precipitation rate, $\hat q^{\text{CESM}}_{i,p}$ is the empirical CESM quantile at location $i$, and $\hat q^{\text{model}}_{i,p}$ is the corresponding empirical quantile of generated model samples. Positive values indicate that the model-sample quantile is higher than the CESM quantile at that location, negative values indicate that it is lower, and values close to zero indicate close marginal agreement at the displayed quantile level.

The two displayed quantile levels were chosen to align with the threshold summaries in the main text. The comparison shows the same qualitative pattern as the averaged diagnostic: the \texttt{SCT-Skew-t} panels are closer to zero over broader regions than the corresponding \texttt{nonlin TM} panels, especially at matched training sizes, while the remaining signed differences retain spatial structure.

The signed maps show the direction of the remaining marginal discrepancies. At $p=0.10$, \texttt{nonlin TM} with $\Ntrain=20$ tends to underestimate the lower quantile in the Southern Hemisphere. \texttt{SCT-Skew-t} reduces this pattern, but with $\Ntrain=20$ it still shows coherent neighboring patches of underestimation and overestimation, especially in the South Atlantic and South Pacific. With $\Ntrain=80$, the \texttt{SCT-Skew-t} map is much closer to zero, apart from small localized residuals such as overestimation near the Himalayas. At $p=0.95$, \texttt{nonlin TM} systematically overestimates upper quantiles, while both \texttt{SCT-Skew-t} fits reduce this substantially. Remaining structure in the \texttt{SCT-Skew-t} maps includes underestimation around the northern Sahara and the Middle East and, for $\Ntrain=20$, overestimation in the equatorial Pacific. These residual patterns do not uniquely identify a mechanism: they may reflect residual marginal smoothing, local overfitting when the training ensemble is small, or missing information from geographical or physical structure.

\begin{figure}[htb]
    \centering
    \begin{minipage}{0.49\linewidth}
    \includegraphics[width=\linewidth]{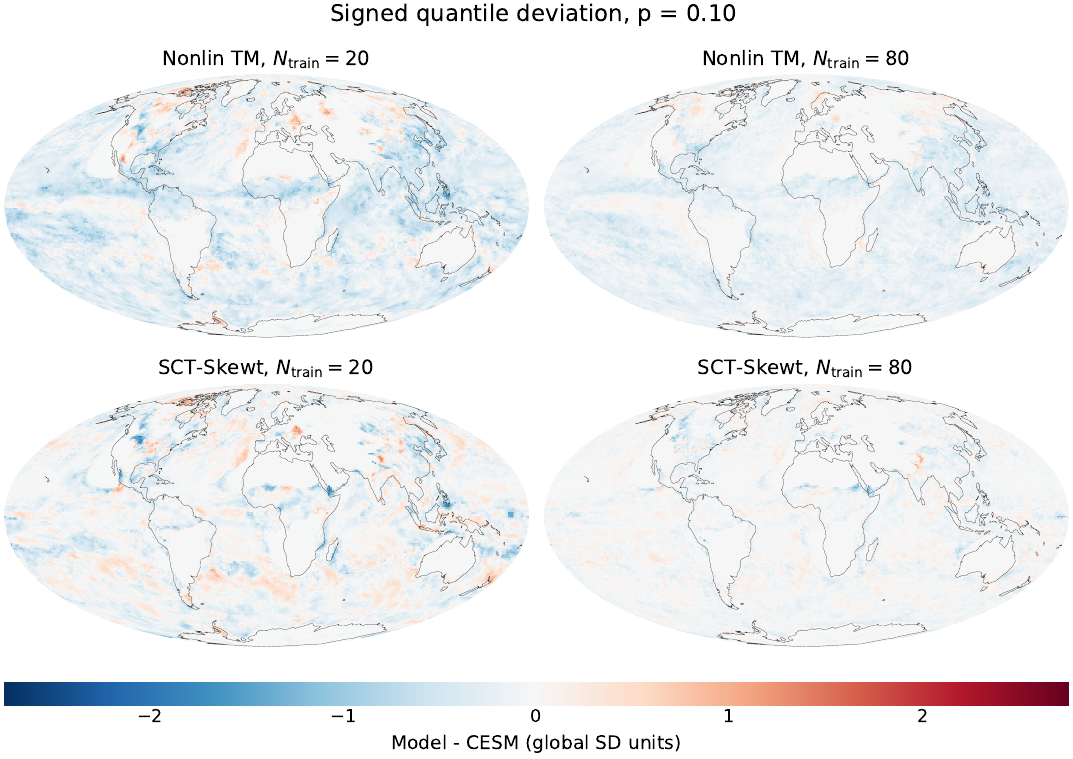}
    \end{minipage}
    \begin{minipage}{0.49\linewidth}
    \includegraphics[width=\linewidth]{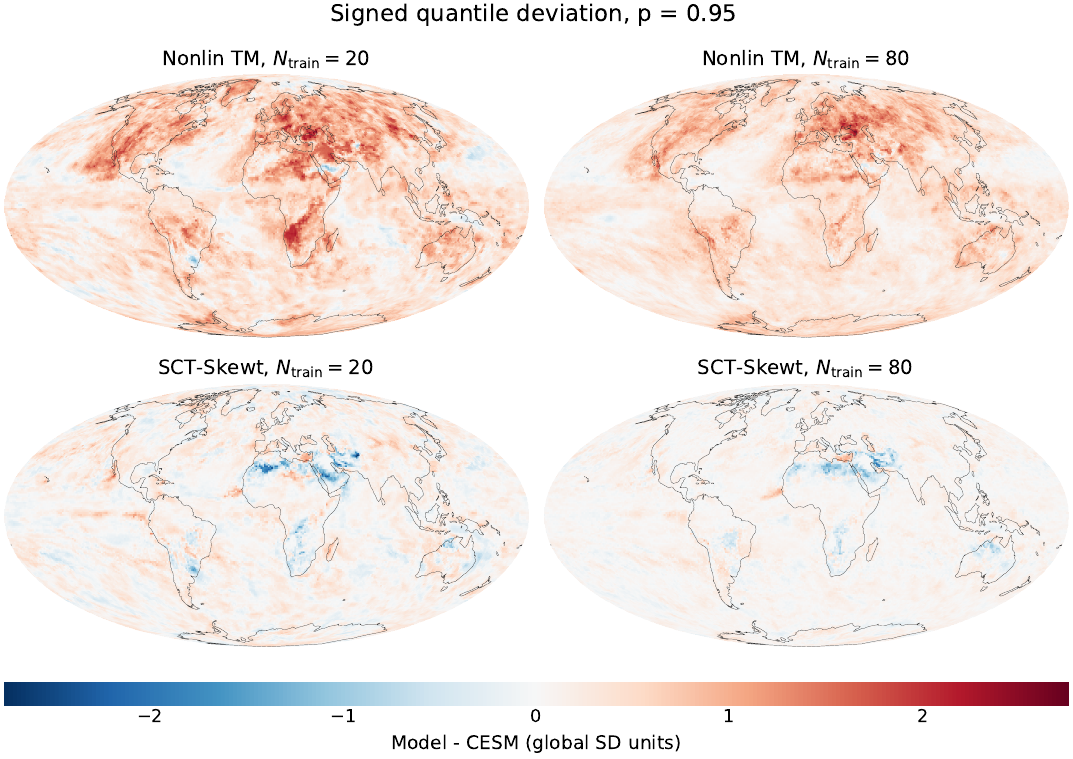}
    \end{minipage}
    
    \caption{Spatially resolved signed marginal quantile differences at $p=0.10$ (left) and $p=0.95$ (right). At each location, the plotted value is the empirical model-sample quantile minus the empirical CESM quantile, scaled by the global sample standard deviation of log precipitation rate. Positive values indicate that the model-sample quantile is higher than the CESM quantile, negative values indicate that it is lower, and values closer to zero indicate closer marginal fit.}
    \label{app-fig:marginal-quantile-diagnostic-levels}
\end{figure}

\FloatBarrier

\clearpage
\renewcommand\bibliographytypesize{\footnotesize}
\bibliography{references}


\end{document}